\documentclass[11pt,a4paper]{article}

\usepackage{graphicx}
\graphicspath{{figs/}}
\usepackage[english]{babel}
\usepackage{latexsym}
\usepackage{url}
\usepackage{amssymb}
\usepackage{amsmath}
\usepackage{amsthm}
\usepackage{amsfonts}
\usepackage{dsfont}
\usepackage{ifthen}
\usepackage{algo}
\usepackage{fancybox}

\usepackage{color}
\definecolor{blu3}{rgb}{.1,.0,.4}
\usepackage{hyperref}
\hypersetup{colorlinks=true, linkcolor=blu3, urlcolor=blu3, citecolor=blu3, pdfpagemode=UseNone, pdfstartview=} 

\usepackage[left=1.2in,top=1.2in,right=1.2in,bottom=1.2in,nohead]{geometry}

\newtheorem{theorem}{Theorem}

\newtheorem{lemma}[theorem]{Lemma}

\newtheorem{remark}[theorem]{Remark}

\newcommand{\RR}{\ensuremath{\mathbb R}}  
\newcommand{\ZZ}{\ensuremath{\mathbb Z}}  
\newcommand{\EE}{\ensuremath{\mathbb E}}  
\newcommand{\E}{\ensuremath{\mathcal{E}}}  
\newcommand{\F}{\ensuremath{\mathcal{F}}}  
\newcommand{\U}{\ensuremath{\mathcal{U}}}  
\DeclareMathOperator{\Var}{Var}
\DeclareMathOperator{\conv}{conv}

\DeclareMathOperator{\area}{area}
\DeclareMathOperator{\per}{per}
\DeclareMathOperator{\dwidth}{dwidth}
\newcommand\eps{\varepsilon}

\newcommand{\piin}[1]{\ensuremath{\Pi_{in}(#1)}}
\newcommand{\piout}[1]{\ensuremath{\Pi_{out}(#1)}}
\newcommand{\scalar}[1]{\ensuremath{\langle #1 \rangle}} 

\newcommand{\OPT}{{\textsc{Opt}}}  

\def\DEF#1{\textbf{\emph{#1}}}

\begin{document}

\title{Peeling potatoes near-optimally in near-linear time\thanks{A preliminary version of this paper appeared in Proc. 30th Annual Symposium on Computational Geometry (SoCG 2014), pp. 224--231.}}

\author{Sergio Cabello\thanks{Department of Mathematics, IMFM, and
                Department of Mathematics, FMF, University of Ljubljana, Slovenia.
                Supported by the Slovenian Research Agency, program P1-0297, projects J1-4106 and L7-5459, and by the ESF EuroGIGA project (project GReGAS) of the European Science Foundation.}
        \and Josef Cibulka\thanks{Department of Applied Mathematics and Institute for Theoretical Computer Science, Charles University, Faculty of Mathematics and Physics, Czech Republic. 
        Supported by the project CE-ITI (GA\v CR P202/12/G061) of the Czech Science Foundation.}
        \and Jan Kyn\v{c}l\thanks{Department of Applied Mathematics and Institute for Theoretical Computer Science, Charles University, Faculty of Mathematics and Physics, Czech Republic; and Alfr\'ed R\'enyi Institute of Mathematics, Hungary. 
        Supported by the project CE-ITI (GA\v CR P202/12/G061) of the Czech Science Foundation and by ERC Advanced Research Grant no 267165 (DISCONV).}
        \and Maria Saumell\thanks{Institute of Computer Science,
                The Czech Academy of Sciences, Czech Republic. With 	  institutional support RVO:67985807.
                Supported by project LO1506 of the Czech Ministry of Education, Youth and Sports, project CE-ITI (GA\v CR P202/12/G061) of the Czech Science Foundation, project NEXLIZ - CZ.1.07/2.3.00/30.0038, co-financed by the European Social Fund and the state budget of the Czech Republic, ESF EuroGIGA project ComPoSe as F.R.S.-FNRS - EUROGIGA NR 13604, and H2020-MSCA-RISE project 73499 - CONNECT.}
        \and Pavel Valtr\thanks{Department of Applied Mathematics and Institute for Theoretical Computer Science, Charles University, Faculty of Mathematics and Physics, Czech Republic. 
        Supported by the project CE-ITI (GA\v CR P202/12/G061) of the Czech Science Foundation.}
}

\date{\today}

\maketitle

\begin{abstract}
	We consider the following geometric optimization problem:
    find a convex polygon of maximum area contained 
    in a given simple polygon $P$ with $n$ vertices.
    We give a randomized near-linear-time $(1-\eps)$-approximation 
    algorithm for this problem:
    in $O(n( \log^2 n + (1/\eps^3) \log n + 1/\eps^4))$ time we find
    a convex polygon contained in $P$ that, with probability at least
    $2/3$, has area at least $(1-\eps)$ times 
    the area of an optimal solution. We also obtain similar results for
    the variant of computing a convex polygon inside $P$ with maximum
    perimeter.
    
    To achieve these results we provide new results 
    in geometric probability. The first result is a bound relating 
    the probability that two points chosen uniformly at random inside $P$ 
    are mutually visible and the area of the largest convex body inside $P$.
	The second result is a bound on the expected value of the difference
    between the perimeter of any planar convex body $K$ and the perimeter
    of the convex hull of a uniform random sample inside $K$.
    
\end{abstract}

\section{Introduction}

We consider the algorithmic problem of finding a maximum-area convex set
in a given simple polygon. Thus, we are interested in computing
\begin{align*}
    A^*(P) ~\mathrel{\mathop:}=~ \sup \{ \area(K) \mid K\subset P,~ K \text{ convex}\}.
\end{align*}
The problem was introduced by Goodman~\cite{Goodman-81}, who named it
the \DEF{potato peeling problem}. Goodman also
showed that the supremum is actually achieved, 
so we can replace it by the maximum.
Henceforth we use $n$ to denote the number of vertices
in the input polygon $P$.

Chang and Yap~\cite{potato-exact} showed that $A^*(P)$ can be computed
in $O(n^7)$ time. Since there have been no improvements 
in the running time of exact algorithms,
it is natural to turn the attention to faster, approximation algorithms.
A step in this direction is made by Hall-Holt et al.~\cite{hkms-06}, who show
how to obtain a constant-factor approximation in $O(n \log n)$ time.

In this paper we present a randomized $(1-\eps)$-approximation algorithm.
Besides the simple polygon $P$, the algorithm takes as input 
a parameter $\eps\in(0,1)$ controlling the approximation.
In time $O\left(n(\log^2 n + (1/\eps^3) \log n + 1/\eps^4)\right)$ the algorithm returns
a convex polygon contained in $P$ that, with probability at least $2/3$, 
has area at least $(1-\eps)\cdot A^*(P)$. 
For any constant $\eps$, and more generally for any
$\eps=\Omega\left(1/\log^{1/3}n\right)$, the running time becomes
$O(n\log^2 n)$. As usual, the probability of error can be reduced to $\delta \in (0,1)$
using $O(\log(1/\delta))$ independent repetitions of the algorithm.
Note that for $\eps<1/n^{3/2}$, the exact algorithm 
of Chang and Yap~\cite{potato-exact} is faster as it runs in time $O(n^7)=O(n/\eps^4)$. 

\paragraph{Overview of the approach.}
Let $R$ be a set of points contained in $P$.
The \DEF{visibility graph} of $R$, denoted by $G(P,R)$, 
has $R$ as vertex set and, for any two points $x$ and $y$ in $R$, 
the edge $xy$ is in $G(P,R)$ whenever the segment
$xy$ is contained in $P$. See Figure~\ref{fig:visibilitygraph}.

Let us assume that the set of points $R$ is obtained by
uniform sampling in $P$. We note the following properties:
\begin{itemize}
    \item For each convex polygon $K\subseteq P$,
        the area of the convex hull $\conv(K\cap R)$ is similar to the area of $K$,
        provided that $|K\cap R|$ is large enough. For this,
        it is convenient to have large $|R|$.
    \item For each convex polygon $K\subseteq P$,
        the boundary of $\conv(K\cap R)$ is made of edges in $G(P,R)$.
    \item With dynamic programming one can find a maximum-area
    	convex polygon defined by edges of $G(P,R)$. 
        For this to be efficient, 
        it is convenient that $G(P,R)$ has few edges.
\end{itemize}
Thus, we have a trade-off on the number of points in $R$ that are needed.
We argue that there is a suitable size for $R$ such that $G(P,R)$
has a near-linear expected number of edges and, with reasonable probability,
the edges of $G(P,R)$ give a good inner approximation to an optimal
solution. Instead of finding the optimal solution directly in $G(P,R)$,
we make a search in a small parallelogram of area $\Theta(A^*(P))$ 
around each edge of $G(P,R)$, performing a second sampling. 
The core of the argument is a bound relating $A^*(P)$ and
the probability that two random points in $P$ are visible.
Such relation was unknown and we believe that it is of independent interest.
See Theorems~\ref{thm:prob1},~\ref{thm:prob2} 
and the follow up work ~\cite{BJVW15} (summarized in Theorem~\ref{thm:BJVW15} here) 
for the precise relations.

\begin{figure}
    \centering
    \includegraphics{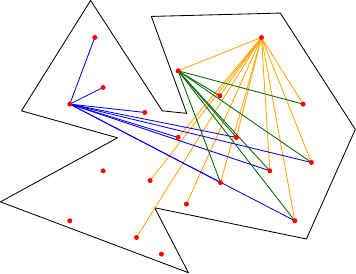}
    \caption{A portion of the visibility graph of a point set. Only the edges incident to
        three vertices are displayed.}
    \label{fig:visibilitygraph}
\end{figure}

\paragraph{Perimeter.} We are also interested in finding a convex polygon inside $P$ 
with maximum perimeter.
Let $\per(K)$ denote the perimeter of a convex body $K$. In the case that
$K$ is a segment, then $\per(K)$ is twice the length of $K$.
Let
\begin{align*}
    L^*(P) ~\mathrel{\mathop:}=~ \sup \{ \per(K) \mid K\subset P,~ K \text{ convex}\}.
\end{align*}
By the same compactness argument as used by Goodman~\cite[Proposition 1]{Goodman-81}, using the Blaschke selection theorem, the supremum is achieved and so it can be replaced by the maximum.

We provide a randomized algorithm to compute a convex polygon (or segment) inside $P$
whose perimeter is at least $(1-\eps)\cdot L^*(P)$. 
For every $\delta > 0$, to succeed with probability $1-\delta$, the algorithm uses time 
\[
O\left( n\left[ (1/\eps^4) \log^2 n + \left((1/\eps) \log^2 n + (1/\eps^6) \log n + 1/\eps^8\right) \log(1/\delta) \right]  \right).
\]
The main obstacle in this case is that the polygons with near-optimal 
perimeter may be very skinny and thus have arbitrarily small area.
For that case, random sampling of points is futile, 
but we can use a longest segment contained in $P$ to approximate $L^*(P)$. 
More precisely, if the perimeter-optimal convex polygon has aspect ratio $O(\eps)$,
then we can $(1-\eps)$-approximate it via a longest segment inside $P$, which
in turn can be $(1-\eps)$-approximated in near-linear time~\cite{hkms-06}.
If the perimeter-optimal polygon has aspect ratio $\Omega(\eps)$, then
it has area at least $\Omega(\eps\cdot A^*(P))$, and the approach based on
random samples of points can be adapted, with a larger number of sample points.
To bound the number of sample points we use a new theorem 
in geometric probability bounding the expected difference between
the perimeter of any planar convex body $K$ and the perimeter of the convex
hull of a random sample inside $K$. See our Theorem~\ref{thm:randomperimeter} for
the precise statement.

\paragraph{Other related work.}
There have been several results about finding maximum-area objects
of certain type inside a given simple polygon. DePano, Ke and 
O'Rourke~\cite{DKO-1987} consider squares and equilateral triangles,
Daniels, Milenkovic and Roth~\cite{DMR-1997} consider axis-parallel 
rectangles, 
Melissaratos and Souvaine~\cite{MS-1992} consider arbitrary triangles.
Subquadratic algorithms to find a longest segment contained 
in a simple polygon were first given by Chazelle and Sharir~\cite{CS-1990}
and improved by Agarwal, Sharir and Toledo~\cite{AS-1996,AST-1994}.
Hall-Holt et al.~\cite{hkms-06} present near-linear time algorithms for a 
$(1-\eps)$-approximation of the longest segment.

Aronov et al.~\cite{meshed-potato} consider a variation where the
search is restricted to convex polygons whose edges are edges
of a given triangulation (with inner points) of $P$. 
They show how to compute a maximum-area convex polygon 
for this model in $O(m^2)$ time, where $m$ is the number of edges
in the triangulation.

Dumitrescu, Har-Peled and T{\'o}th~\cite{Dumitrescu-et-al} consider 
the following problem: given a unit square $Q$ and a set $X$ 
of points inside $Q$, find a maximum-area convex body inside $Q$ 
that does not have any point of $X$ in its interior. This is
an instance of the potato peeling problem for polygons \emph{with holes}.
They provide a $(1-\eps)$-approximation
in time $O(n^2/\eps^6)$. For any fixed $\eps$, 
the running time is quadratic. Our algorithm exploits the
absence of holes in $P$, so it does not produce an improvement
in this case.

The potato peeling problem can be understood as finding a largest 
set of points that are mutually visible. 
Rote~\cite{degree-convexity} showed how to compute in polynomial time
the probability that two random points inside a polygon are visible. A faster algorithm has been proposed by Buchin et al.~\cite{Buchin-visib}.
Cheong, Efrat and Har-Peled~\cite{shop} consider the problem of finding
a point in a simple polygon whose visibility region is maximized.
They provide a 
$(1-\eps)$-approximation algorithm using near-quadratic time. The approach
is based on taking a random sample of points in the polygon, 
constructing the visibility region of each point, 
and taking a point lying in most visibility regions.

\paragraph{Roadmap.} In Section~\ref{sec:convex} 
we provide tools related to convex bodies.
In Section~\ref{sec:probability} we relate the probability of 
two random points being visible and $A^*(P)$.
We present and analyze the algorithm to approximate $A^*(P)$ in 
Section~\ref{sec:algorithm}. 
In Section~\ref{sec:perimeter} we discuss the adaptation to maximize the perimeter.
We conclude in Section~\ref{sec:conclusions}.

\paragraph{Assumptions.} We will have to generate
points uniformly at random inside a triangle. For this, 
we will assume that a random number in the interval $[0,1]$ 
can be generated in constant time.

\section{About convexity}
\label{sec:convex}

Here we provide tools related to convexity.

\subsection{Inner approximation using random sampling}

In this subsection, we provide results about the number of points that
have to be sampled inside a convex body $K$ 
so that the area of the convex hull of the sample 
is a good approximation to the area of $K$. We may think of $K$ as a maximum-area convex set in $P$ for which we aim to find a  
$(1-\eps)$-approximation. In our algorithm, we sample points
in a superset of $K$, thus
we also provide extensions to this case. In particular, Lemma~\ref{le:convex} deals with the problem of sampling points inside a given convex body $K$. In Lemma~\ref{le:sampleinside} the sample is taken from a larger polygon $\Gamma\supseteq K$ and the goal is to hit $K$ with at least $C$ points. These two results are then combined together in Lemma~\ref{le:sampleinside2}.

\begin{lemma}
\label{le:convex}
    Let $K$ be a convex body in the plane and let $R$ 
    be a sample of points chosen uniformly at random inside $K$. 
    There is some universal constant $C_1$ such that,
    if $|R|\ge C_1/ \eps^{3/2}$, 
    then with probability at least $5/6$
    it holds that $\area(\conv(R))\ge (1-\eps)\cdot \area(K)$.
\end{lemma}
\begin{proof}
    We use as a black box known extremal properties and
    bounds on the so-called missed area of a random polygon.
    See the lectures by B\'ar\'any~\cite[2nd lecture]{stochasticgeometry}, the survey~\cite{Ba08} or
    \cite{BL88} for an overview.

    Let us scale $K$ so that it has area $1$. 
    We have to show that $1-\area(\conv(R))\ge \eps$
    holds with probability at most $1/6$.

    Let $K_m$ denote the convex hull of $m$ points 
    chosen uniformly at random in $K$
    and define $X(m) = 1- \area(K_m)$. 
    Thus $X(m)$ is the \emph{missed area}, that is, 
    the area of $K\setminus K_m$.
    Groemer~\cite{Groemer} showed that $\EE [X(m)]$ 
    is maximized when $K$ is a disk of area $1$.
    R\'enyi and Sulanke~\cite{RS64_II} showed that 
    for every smooth convex set $K$ there exists
    some constant $C_K$, depending on $K$, such that
    $\EE [X(m)]\leq C_K \cdot m^{-2/3}$. This result also follows from a similar upper bound by R\'enyi and Sulanke~\cite{RS63} on the expected number $E_m$ of edges of $K_m$ and from Efron's~\cite{Efron65} identity $\EE [X(m)]= \EE [E_{m+1}]/(m+1)$.
    Both statements together imply that
    \[
	    \EE [X(m)] \le \frac{C'}{m^{2/3}},
    \]
    where $C'$ is the constant $C_K$ when $K$ is a unit-area disk.
    (From the results of~\cite{RS64_II}, or subsequent works,
    one can explicitly compute that $C' \le 5$, 
    so the constant is very reasonable.)

    We set $C_1 \mathrel{\mathop:}= (6C')^{3/2}$.
    Whenever $|R| \ge C_1 \cdot \eps^{-3/2}$, 
    we can use Markov's inequality to obtain
   
\belowdisplayskip=-8pt
    \begin{align*}
        \Pr [ 1-\area(\conv(R)) \ge \eps] ~& =~ \Pr [X(|R|) \ge \eps] \\
                                       & \le~  \frac{\EE [X(|R|)]}{\eps}\\
                                       & \le~  \frac{C'\, |R|^{-2/3}}{\eps}\\
                                       & \le~  \frac{C'\, \left( (6C')^{3/2} \cdot \eps^{-3/2}\right)^{-2/3}}{\eps}\\
                                       & =~ \frac{1}{6} \, .
    \end{align*}
\qedhere
\end{proof}

\begin{remark}\label{re:C_1}
For convenience we will assume that $C_1/\eps^{3/2}\ge 3$ 
for all $\eps\in (0,1)$. 
This is not problematic because we can replace $C_1$ with $\max(C_1,3)$,
if needed.
\end{remark}

\begin{lemma}\label{le:sampleinside}
	Let $K$ be a convex body contained in a polygon $\Gamma$,
    let $R$ be a random sample of points inside $\Gamma$,
    and let $C\ge 3$ be an arbitrary value.
    If 
    \[
    	|R|~\ge~ 4\cdot C \cdot \frac{\area(\Gamma)}{\area(K)}\,,
    \] 
    then with probability at least $5/6$ it holds that
    $|R\cap K|\ge C$.
\end{lemma}
\begin{proof}
	Let $X=|R\cap K|$.
    The random variable $X$ is a sum of $|R|$ 
    independent Bernoulli random variables, each
    with expected value
    \[
    	p~=~ \frac{\area(K)}{\area(\Gamma)}\, .
    \]
    Standard calculations (or formulas) show that
    \[
    	\EE[X] ~=~ |R|\cdot p ~\ge~ 
        	4\cdot C \cdot 
            \frac{\area(\Gamma)}{\area(K)}\cdot \frac{\area(K)}{\area(\Gamma)}
            ~=~ 4\cdot C
    \]   
    and
	\[     
        \Var[X] ~=~ |R| \cdot p(1-p) ~\le~ \EE[X].
    \]
	We can now use Chebyshev's inequality in its form
    \[
    	\forall a>0: ~~~ \Pr[|X-\EE[X]| \ge a ] ~\le~ \frac{\Var[X]}{a^2} 
    \]
	and the inequality $C\ge 3$ to obtain the following:
\belowdisplayskip=-8pt
    \begin{align*}
    	\Pr[ X \le C] ~&\le~ 
        	\Pr\left[ X\le \tfrac 14 \EE[X] \right]\\
            &\le~ \Pr\left[ |X-\EE[X]|\ge \tfrac 34 \EE[X] \right] \\
            &\le~ \frac{4^2}{3^2} \cdot \frac{\Var[X]}{(\EE[X])^2}\\
            &\le~ \frac{16}{9} \cdot \frac{1}{\EE[X]}\\
            &\le~ \frac{16}{9} \cdot \frac{1}{4\cdot C}\\
            &\le~ \frac{16}{9} \cdot \frac{1}{4\cdot 3}\\
            & < ~ \frac{1}{6}\, .
    \end{align*}
\qedhere    
\end{proof}

\begin{lemma}\label{le:sampleinside2}
	Let $K$ be a convex body contained in a polygon $\Gamma$,
    let $R$ be a random sample of points inside $\Gamma$,
    and let $C_1$ be the constant in Lemma~\ref{le:convex}.
    If 
    \[
    	|R|~\ge~ 4\cdot \frac{C_1}{\eps^{3/2}} \cdot \frac{\area(\Gamma)}{\area(K)}\,,
    \] 
    then with probability at least $2/3$ it holds that
    $\area(\conv(R\cap K))\ge (1-\eps)\area(K)$.
\end{lemma}
\begin{proof}
	We define the following events:
    \begin{align*}
       \E:&~~~ |R\cap K|~\ge~ C_1/\eps^{3/2},\\
       \F:&~~~ \area(\conv(R\cap K))~\ge~ (1-\eps)\cdot \area(K) .   
    \end{align*}
    For each event $\mathcal{A}$ we use $\overline{\mathcal{A}}$ 
    for its negation.
    Since $C_1/\eps^{3/2}\ge 3$ (see Remark~\ref{re:C_1}),
    Lemma~\ref{le:sampleinside} implies
    \[
    	\Pr\left[ \overline{\E} \right] ~\le~ \frac 16\, ,
    \]
    and Lemma~\ref{le:convex} implies
    \[
    	\Pr\left[ \overline{\F} \mid \E \right] ~\le~ \frac 16\, .
    \]
    Therefore
\belowdisplayskip=-8pt
    \begin{align*}
    	\Pr\left[ \overline{\F}\right] ~&=~ 
         \Pr\left[ \overline{\F} \mid \E \right]\cdot \Pr\left[\E\right]
         + \Pr\left[ \overline{\F} \mid \overline{\E} \right]\cdot 
         	\Pr\left[ \overline{\E}\right]\\
         &\le~ \Pr\left[ \overline{\F} \mid \E \right] + 
        		\Pr\left[ \overline{\E}\right]\\
         &\le~ \frac 16 +\frac 16\\
         &= \frac 13\, .
    \end{align*}
\qedhere
\end{proof}

\subsection{Outer containment in a parallelogram}
\label{sec:gamma}

In the previous subsection we have proved that, given a superset $\Gamma$ of $K$, for samples $R$ of points in $\Gamma$ of a certain size, $\area(\conv(R\cap K))$
is a good approximation to the area of~$K$ with positive constant probability. If we set $\Gamma:=P$, the size of $R$ might turn out too big to yield a subquadratic algorithm. For this reason, we want to find a smaller superset of $K$ to take the sample from.
In this subsection we show a method to find a parallelogram $\Gamma$ containing $K$ with area proportional to the area of $K$, and that this parallelogram can be found with positive constant probability, using a relatively small random sample from $P$. 

Let $K$ be a convex body in $\RR^2$. We use $y(p)$ to denote
the $y$-coordinate of a point $p$.
For each $\alpha \in (0,1)$ we define $y_\alpha(K)$ as the unique
value satisfying 
\[
	\area \left( \{ p\in K\mid y(p)\le y_\alpha(K) \} \right) ~=~ 
    	\alpha\cdot \area(K).
\]
Thus, the horizontal line at height $y_\alpha(K)$ breaks $K$ into
two parts and the lower one has a proportion $\alpha$ of the area
of $K$.
We further define
\[
    y_0(K)~\mathrel{\mathop:}=~ \min \{ y(p)\mid p\in K\} ~\text{ and }~~
    y_1(K)~\mathrel{\mathop:}=~ \max \{ y(p)\mid p\in K\}.
\]

\begin{figure}
    \centering
    \includegraphics{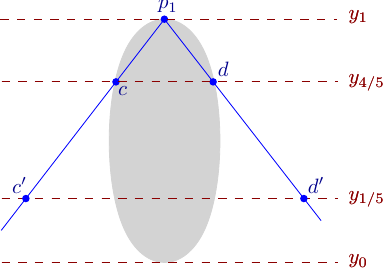}
    \caption{Proof of Lemma~\ref{le:levels}.}
    \label{fig:levels}
\end{figure}

\begin{lemma}
\label{le:levels}
	For each convex body $K$ in $\RR^2$
    \[
    	y_1(K)-y_{4/5}(K)\le y_{4/5}(K)-y_{1/5}(K) ~\text{ and }~ 
        y_{1/5}(K)-y_0(K)\le y_{4/5}(K)-y_{1/5}(K).
    \]
\end{lemma}
\begin{proof}
	In this proof, let us drop the dependency on $K$ in the notation
    and set $y_\alpha \mathrel{\mathop:}= y_\alpha(K)$ for each $\alpha\in [0,1]$. 
    We only show that $y_1-y_{4/5}\le y_{4/5}-y_{1/5}$; the other
    inequality is symmetric.
    
    For $\alpha \in [0,1]$, let $\ell_{\alpha}$ be the horizontal line with $y$-coordinate $y_{\alpha}$. Let $p_1$ be a highest point of $K$,
    let $cd$ be the intersection of $\ell_{4/5}$ with $K$,
    let $\ell_c$ be the line through $p_1$ and $c$, 
    let $\ell_d$ be the line through $p_1$ and $d$,
    let $c'$ be the intersection of $\ell_c$ with $\ell_{1/5}$, and 
    let $d'$ be the intersection of $\ell_d$ with $\ell_{1/5}$.
    See Figure~\ref{fig:levels}.
   	
    By the convexity of $K$, the triangle $p_1cd$ is 
    contained in the portion of $K$ between $\ell_{1}$ and $\ell_{4/5}$, and the portion of $K$ between $\ell_{4/5}$ 
    and $\ell_{1/5}$ is contained in the trapezoid $cc'd'd$. 
    Thus
    \[
    	\area (cc'd'd) \ge 3 \area (p_1cd)
    \]
    and
    \begin{equation}
    \label{eq:area-fraction}
    	\area (p_1c'd') = \area(cc'd'd) + \area (p_1cd) \ge 4 \area (p_1cd).
    \end{equation}
    
    The triangle $p_1c'd'$ is similar to the triangle $p_1cd$ with scale
    factor $(y_{1}-y_{1/5})/(y_{1}-y_{4/5})$.
    By~\eqref{eq:area-fraction}, the scale factor is at least $2$, that is
    \[
    y_{1}-y_{1/5} \ge 2 (y_{1}-y_{4/5})
    \]
    and so
    \[
    y_{4/5}-y_{1/5} \ge y_{1}-y_{4/5}. \qedhere
    \]
\end{proof}

For any two points $a$ and $b$ and any value $A\ge 0$, let $a'\mathrel{\mathop:}=2a-b$, $b'\mathrel{\mathop:}=2b-a$, and
let $\Gamma(a,b,A)$ denote the parallelogram whose vertices are
the four horizontal translates of the points $a'$ and $b'$ by distance $\frac{2A}{|y(a)-y(b)|}$.
See Figure~\ref{fig:gamma}, left.
Note that $|a'b'|= 3|ab|$ and $\area(\Gamma(a,b,A))=12\cdot A$.

\begin{figure}
    \centering
    \includegraphics[width=\textwidth]{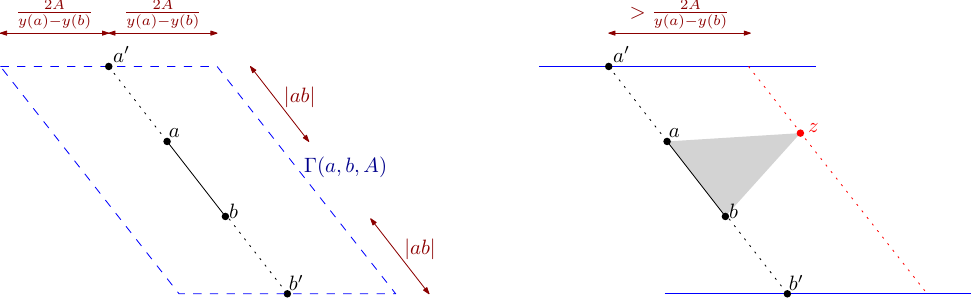}
    \caption{Left: parallelogram $\Gamma(a,b,A)$.
    		Right: proof of Lemma~\ref{le:gamma}.}
    \label{fig:gamma}
\end{figure}

\begin{lemma}
\label{le:gamma}
	Let $K$ be a convex body and assume that $A\ge \area(K)$.
	Let $a$ and $b$ be points in $K$ such that
    \[
    	y(a)\ge y_{4/5}(K) ~\text{ and }~ y(b)\le y_{1/5}(K).
    \]
    Then $K$ is contained in $\Gamma(a,b,A)$.
\end{lemma}
\begin{proof}
	In this proof, let us drop the dependency on $K$ in the notation
    and set $y_\alpha\mathrel{\mathop:}=y_\alpha(K)$ for each $\alpha\in [0,1]$.
    
    By Lemma~\ref{le:levels} we have
    \[
    	y(a') ~=~ 2y(a)-y(b) ~\ge~ y(a) + y_{4/5} - y_{1/5} ~\ge~ 
        	y(a) + y_{1} - y_{4/5} \ge y_1
    \]
    and similarly
    \[
    	y(b') ~=~ 2 y(b)-y(a) ~\le~ 
        	y(b) + y_{1/5} - y_{4/5} ~\le~ y(b)+ y_0 - y_{1/5} ~\le~ y_0.
    \]
    Therefore $K$ is contained between the horizontal lines
    $y=y(a')$ and $y=y(b')$. These are the lines supporting the top and
    bottom side of $\Gamma(a,b,A)$.
    
    Assume, for the sake of a contradiction, that $K$ has some point 
    $z$ outside $\Gamma(a,b,A)$. Since $z$ lies between the lines $y=y(b')$
    and $y=y(a')$, it must be that the horizontal distance from $z$ to $a'b'$ 
    is more than $\frac{2A}{y(a)-y(b)}$. 
    See Figure~\ref{fig:gamma}, right.
    Since the triangle $abz$ is
	contained in $K$ we would have 
    \[
    	\area(K) ~\ge~ \area(abz) ~>~ 
        	\frac 12 \cdot (y(a)-y(b)) \cdot \frac{2A}{y(a)-y(b)}
        ~=~ A ~\ge~ \area(K), 
    \]
    which is a contradiction.
    Therefore any point of $K$ is contained in $\Gamma(a,b,A)$.    
\end{proof}

\begin{lemma}\label{le:samplegamma}
	Let $K$ be a convex body contained in a polygon $P$,
    and assume that $A\ge \area(K)$.
    If $R$ is a random sample of points inside $P$ with
    \[
    	|R|~\ge~ 60\cdot \frac{\area(P)}{\area(K)}\,,
    \] 
    then with probability at least $2/3$ it holds that
    $R$ contains two points $a$ and $b$ such that $ab$ is an edge of $G(P,R)$ and
    $\Gamma(a,b,A)$ contains $K$.
\end{lemma}
\begin{proof}
	Define 
    \[
    	K_{\le 1/5} \mathrel{\mathop:}= \{ p\in K\mid y(p)\le y_{1/5}(K)\} ~~\text{ and }~~
        K_{\ge 4/5} \mathrel{\mathop:}= \{ p\in K\mid y(p)\ge y_{4/5}(K)\},
    \]
	and consider the following events:
    \[
       \E_{\le 1/5}:~ K_{\le 1/5}\cap R \not= \emptyset~~\text{ and }~~
       \E_{\ge 4/5}:~ K_{\ge 4/5}\cap R \not= \emptyset .   
    \]
    Since 
    \[
    	|R|~\ge~ 4\cdot 3\cdot \frac{\area(P)}{\area(K)/5} ~=~
        4\cdot 3\cdot \frac{\area(P)}{\area(K_{\le 1/5})} ~=~
         4\cdot 3\cdot \frac{\area(P)}{\area(K_{\ge 4/5})} \,,
    \] 
    Lemma~\ref{le:sampleinside} implies
    \[
    	\Pr\left[ \E_{\le 1/5} \right] ~\ge~ \frac 56 ~~\text{ and }~~
        \Pr\left[ \E_{\ge 4/5} \right] ~\ge~ \frac 56\,.
    \]
    Applying the Fr\'echet inequality 
    \[
    	\Pr\left[ A \cap B \right]\ge \max \{0,\Pr\left[ A\right]+\Pr\left[ B\right]-1\}\,,
    \] 
    which does not require any independence assumption, we obtain that 
    \[
    	\Pr\left[ \E_{\le 1/5} \cap \E_{\ge 4/5} \right] ~\ge~ \frac 23\, .
    \]
    When $\E_{\le 1/5}$ and $\E_{\ge 4/5}$ hold, there are points
    $a\in K_{\le 1/5}\cap R$ and $b\in K_{\ge 4/5}\cap R$ and
    Lemma~\ref{le:gamma} implies that $K$ is contained in $\Gamma(a,b,A)$.       Moreover, $ab$ is an edge of $G(P,R)$ because $K$ is a convex body contained in $P$.
\end{proof}

\subsection{Largest convex polygon in a visibility graph.}
\label{sec:phi}

In this subsection we give an algorithm to find a largest convex
polygon whose edges are defined by a visibility graph inside a polygon. In our algorithm {\sc LargePotato}, described in Section~\ref{sec:algorithm}, the vertices of the visibility graph are points of a random sample in $P$, and the algorithm in the current subsection is used to find the largest convex polygon defined by that sample.

Let $H$ be a visibility graph in some simple polygon. We denote the set of vertices and edges of $H$ by $V(H)$ and $E(H)$, respectively.
We assume that the coordinates of the vertices of $H$ are known.
A set of vertices $U$ from $H$ is a \DEF{convex clique} if:
(i) there is an edge between any two vertices of $U$, and
(ii) the points of $U$ are in convex position.
The \DEF{area of a convex clique} $U$ is the area of $\conv(U)$.

Let $s$ be a point of $V(H)$. 
We are interested in finding a convex clique of maximum area in $H$, denoted by $ \varphi(H,s)$,
that has $s$ as highest point.
Thus we want
\[
    \varphi(H,s) ~\in~ \arg\max \{ \area(U) \mid
                \text{$U\subseteq V(H)$ a convex clique, 
                		$s$ highest point in $U$}\}.
\]

\begin{lemma}
\label{le:phi}
    For any point $s$ of $V(H)$, we can compute $\varphi(H,s)$
    in time $O(|V(H)|^2)$.
\end{lemma}
\begin{proof}
    Pruning vertices, we can assume that all vertices of $H$ 
    are adjacent to $s$ and below $s$.
    We can then use the algorithm of Bautista-Santiago et al.~\cite{islands}, which is an improvement
    over the algorithm of Fischer~\cite{Fischer}, restricted to the edges that are in $H$.
    For completeness, we provide a quick overview of the approach.

    For this proof, let us denote $n=|V(H)|-1$.
    We sort the points of $V(H)\setminus \{ s\}$ counterclockwise radially from $s$.
    Let $x_1,x_2,\dots, x_n$ be the labeling of the points of $V(H)\setminus \{ s\}$
    according to that ordering.
    Thus, for each $i<j$ the sequence $x_i, s, x_j$ is a right turn.

    Using a standard point-line duality and constructing the arrangement of lines dual
    to the points $V(H)$, we get the circular order of the edges around each point $x_i$~\cite{LC85_rotations_On2}.
    For this we spend in total $O(n^2)$ time~\cite{CGL85_arrangement_On2, EORS86_arrangement_On2}.

    For each $i<j$ such that $x_ix_j\in E(H)$,
    let $\OPT[i,j]$ be the largest-area convex clique $U$
    that has $x_i$, $x_j$, and $s$ \emph{consecutively} along the boundary of $\conv(U)$.
    We then have
    \[
        \area(\varphi(H,s))= \max_{i<j, x_ix_j\in E(H)} \OPT[i,j].
    \]
    Taking the convention that $\max \emptyset = 0$,
    the values $\OPT[i,j]$ satisfy the following recursion
    \begin{align*}
        \OPT[i,j] ~=~ \area&(sx_ix_j) \\
                    &+~ \max \{ \OPT[h,i] \mid h<i,\, x_h x_i\in E(H), \text{ $x_h,x_i,x_j$ makes a left turn}\}.
    \end{align*}
    To argue the correctness of the recursion, one needs to observe that the right side
    of the equation does indeed correspond to the construction of a convex polygon.

    For any fixed $i$, the values $\OPT[i,*]$, $*>i$, can be computed in $O(n)$ time,
    provided that the edges incident to $x_i$ are already radially sorted and the values $\OPT[h,i]$
    are already available for all $h<i$. To achieve linear time,
    one performs a scan of the edges incident to $x_i$ and uses the property that
    \[
        \{ x_h x_i\in E(H) \mid h<i, \text{ $x_h,x_i,x_j$ makes a left turn}\}
    \]
    forms a contiguous sequence in the circular ordering of edges incident to $x_i$.

    Thus, we can fill in the whole table $\OPT[\cdot,\cdot]$ in time $O(n^2)$.
    With this we can compute $\area(\varphi(H,s))$ and construct an optimal
    solution $\varphi(H,s)$ by standard backtracking.
    See~\cite{islands} for additional details.
\end{proof}

In Section~\ref{sec:perimeter} we will also need to find a convex clique $U$ whose convex hull has maximum perimeter. It is easy to modify the algorithm to compute, for a point $s\in V(H)$, the value
\[
    \max \{ \per(U) \mid
                \text{$U\subseteq V(H)$ a convex clique, 
                		$s$ highest point in $U$}\}
\]
and a corresponding optimal solution. Here we assume a model of computation where the length of segments can be added in constant time.

\section{Probability for visibility}
\label{sec:probability}

In this section we give a relation between $A^*(P)$ and
the probability that two random points in $P$ are visible. Such a relation is used later to bound the expected complexity of the visibility graph of a suitably sized random sample of points.

A polygon $P$ is \DEF{weakly visible} from a segment $s$ in $P$ if, 
for each point $p\in P$, there exists some point $x\in s$ 
such that $xp\subset P$.

\begin{theorem}\label{thm:prob1}
    Let $P$ be a unit-area polygon weakly visible from a diagonal $s$.
    Let $a$ and $b$ be two points chosen uniformly at random in $P$.
    Then
    \begin{itemize}
    \item[\textup(i\textup)] $ \Pr \left[ ab \subset P \right] ~\le~ 18\cdot  A^*(P)$ and\footnote{Item (i) is not used elsewhere in this paper. However, we believe that it is an interesting fact that strengthens Theorem~\ref{thm:prob2} for weakly edge-visible polygons.}
    \item[\textup(ii\textup)] $\Pr \left[ ab \subset P \textup{ and } ab \cap s \neq \emptyset \right] ~\le~ 6\cdot  A^*(P)$.
    \end{itemize}
\end{theorem}
\begin{proof}
    Without loss
    of generality we assume that $s$ is a horizontal segment on the $x$-axis. 
    In this proof we use $y(a)$ to denote the $y$-coordinate of a point $a$. Since the event that $y(a)y(b)=0$ has zero probability, we may assume that $y(a)\neq 0$ and $y(b)\neq 0$.
    To simplify the notation, in this proof we use $A^*=A^*(P)$.
    
    Consider first the point $a$ fixed.
    We first bound the probability that $a$ and $b$ are visible
    and $|y(a)|\ge |y(b)|$ to obtain the following:
    \[
        \Pr \left[ ab \subset P \text{ and }
        		|y(a)|\ge |y(b)| \right] 
        	~\le~ 9\cdot  A^*.
    \]
    This is seen showing that the set of points $b$
    satisfying ${ab} \subset P$ and $|y(a)|\ge |y(b)|$
    is inside a region of area at most $9A^*$.
    
    We distinguish two cases:
    \begin{itemize}
    	\item[1)] $y(a)y(b) > 0$  ($a$ and $b$ are on the same side of $s$).
    	\item[2)] $y(a)y(b) < 0$   ($a$ and $b$ are on the opposite sides of $s$).
	\end{itemize}

\begin{figure}
    \centering
    \includegraphics{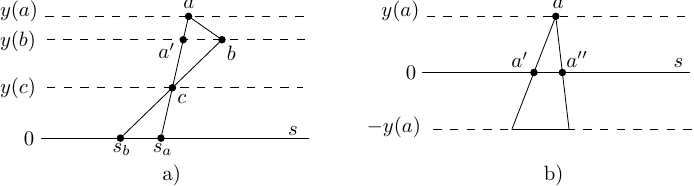}
    \caption{Situation in the proof of Theorem~\ref{thm:prob1}. 
    a) Points $a$ and $b$ are on the same side of $s$. 
    b) Points $a$ and $b$ are on different sides of $s$.}
    \label{fig:probBound}
\end{figure}

	Let us first consider case 1). We assume that $y(a)>0$, the other case is symmetric.
	Refer to Figure~\ref{fig:probBound}a).
    We know that $a$ sees some point $s_a$ on $s$.
    We may assume that $b$ does not lie on the segment $as_a$ as the event that $b$ lies on $as_a$ has zero probability.
    We know that $b$ sees some point $s_b$ on $s$.
    We have a generalized polygon $Q=abs_bs_a$ (in which the sides $as_a$ and $bs_b$ may cross or some of the vertices may coincide) 
    whose boundary is in $P$, 
    and therefore the whole interior of $Q$ is also in $P$. Here we use that $P$ has no holes.
If $b$ sees $a$, we can choose $s_b$ so that
    $as_a$ and $bs_b$ share a common point: indeed, if $as_a$ and $bs_b$ are disjoint, then the polygon $Q$ is simple and thus $b$ sees $s_a$, so we can set $s_b$ to $s_a$.
    Let $c$ be the common point of $as_a$ and $bs_b$. By our assumptions, $y(c)<y(a)$.
        
    Let $h$ be a horizontal line through $b$ and let $a'$ be the intersection
    between $h$ and the segment $as_a$. 
    The interior of $Q$ is made of two triangles, $abc$ and $s_a s_b c$, 
    both contained in $P$ and thus each of them has area at most $A^*$. The triangle $s_a s_b c$ degenerates to a point if $s_a = s_b$.
    
    For the triangle $abc$, we have $\area(abc) ~=~ \tfrac12 |a'b| \cdot (y(a)-y(c))$, which implies that 
    \begin{equation}\label{eq_1}
    |a'b|\le \frac{2A^*}{y(a)-y(c)}.
    \end{equation}
    If the triangle $s_a s_b c$ is not degenerate, we have $y(c)\cdot |s_as_b| \le 2A^*$. By the similarity of the triangles $s_a s_b c$ and $a'bc$, we have $|s_as_b|=|a'b|\cdot y(c) / (y(b)-y(c))$, which implies that
    \begin{equation}\label{eq_2}
    |a'b|\le \frac{2A^*}{y(c)^2}\cdot(y(b)-y(c)).
    \end{equation}
    Since the upper bound on $|a'b|$ is increasing in $y(c)$ in~(\ref{eq_1}) and decreasing in $y(c)$ in~(\ref{eq_2}), the minimum of the two upper bounds is maximal when they are equal; that is, when $y(c)=y(a)y(b) / (y(a)+y(b))$. It follows that
    \begin{equation}\label{eq_3}
    |a'b|\le \frac{2A^*}{y(a)^2}\cdot(y(a)+y(b)).
    \end{equation}
The condition~(\ref{eq_3}) implies that $b$ is inside a trapezoid of height $y(a)$ with bases of length $4 A^*/y(a)$ and $8 A^*/y(a)$, which has area $6A^*$. This finishes case 1).
    %
    %
    %

	We now consider case 2). 
	Refer to Figure~\ref{fig:probBound}b).
    Let $a'a''$ be the maximum subsegment of $s$ that is
    visible from $a$. Since the triangle $aa'a''$ is contained
    in $P$ we have
    \[
    	\area(aa'a'') ~=~ \tfrac12 |a'a''| \cdot y(a) ~\le~ A^*.
    \] 
    If $b$ sees $a$, then the segment $ab$ intersects the segment $a'a''$.
    Thus $b$ is contained in a trapezoid of height $y(a)$
    with bases of length $|a'a''|$ and $2\cdot |a'a''|$.
    Such trapezoid has area
    \[
    	\frac{|a'a''|+ 2|a'a''|}{2} \cdot y(a) ~=~ \tfrac32 |a'a''| \cdot y(a)
        ~\le~ 3A^*.
    \]
    This finishes case 2).
    
	Considering cases 1) and 2) together,
    for each fixed point $a\in P$ we have 
    \[
       \Pr \left[ ab \subset P \text{ and }
             |y(a)|\ge |y(b)| \right] 
       ~\le~ 9\cdot  A^*.
     \]
     Since this bound holds for each fixed $a$, it
     also holds when $a$ is chosen at random.
     
     Because of symmetry we have
     \[
       \Pr \left[ ab \subset P\right] ~=~
       		2\cdot \Pr \left[ ab \subset P \text{ and } |y(a)|\ge |y(b)| \right]~\le~ 18\cdot  A^*,
     \]
     which proves part (i) of the theorem.

     Part (ii) follows by a similar consideration using case 2) only.
\end{proof}

We can use a divide and conquer approach to obtain a bound for arbitrary polygons.
\begin{theorem}\label{thm:prob2}
    Let $P$ be an arbitrary unit-area polygon.
    Let $a$ and $b$ be two points chosen uniformly at random in $P$.
    Then
    \[
        \Pr \left[ {ab} \subset P \right] ~\le~ 
        		12\cdot  A^*(P) \cdot \bigl( 1+ \log_2 (1/A^*(P)) \bigr).
    \]
\end{theorem}
\begin{proof}
	For this proof, let us set $A^*=A^*(P)$.
    
	For each polygon $Q$ there exists a segment that splits $Q$
    into two polygons, each of area at most $\tfrac23 \area(Q)$~\cite{polygon-cutting}.
    We recursively split $P$ using such a segment in each polygon, for
    $h=\log_{3/2} (1/A^*)$ levels.    
    Thus, at the bottommost level,
    each polygon has area bounded by $A^*$.
    
    At each level $\ell$ of the recursion, where $\ell=0,\dots,h$, 
    we have $2^\ell$ polygons, 
    which we denote by $Q_{\ell,1}, \dots, Q_{\ell, 2^\ell}$.
    In particular, $Q_{0,1} =P$. Since the polygons at each level $\ell$
    are disjoint, we have
    \[
    	\sum_{i=1}^{2^\ell} \area(Q_{\ell,i}) ~=~ \area(P) ~=~ 1.
    \] 
    For each polygon $Q_{\ell,i}$, where $\ell<h$, 
    let $e_{\ell,i}$ be the segment used to split $Q_{\ell,i}$. 
    Let $\widehat{Q_{\ell,i}}$ be the portion of $Q_{\ell,i}$ that is weakly
    visible from $e_{\ell,i}$. 
    At each level $\ell<h$ we have
    \[
    	\sum_{i=1}^{2^\ell} \area(\widehat{Q_{\ell,i}}) 
        ~\le~ \sum_{i=1}^{2^\ell} \area(Q_{\ell,i}) 
        ~=~ 1.
    \]
    Let $\E_{a,b,\ell,i}$ be the event 
    $ab \subset \widehat{Q_{\ell,i}} \text{ and } ab \cap e_{\ell,i} \neq \emptyset$.
    Using the union bound and part (ii) of Theorem~\ref{thm:prob1} we obtain
    \begin{align*}
        \Pr \left[ \bigcup_{\ell=0}^{h-1}\bigcup_{i=1}^{2^\ell} 
        				\E_{a,b,\ell,i} \right] 
        ~&\le~ \sum_{\ell=0}^{h-1} \Pr \left[ 
        		\bigcup_{i=1}^{2^\ell} \E_{a,b,\ell,i} \right] \\
        &=~ \sum_{\ell=0}^{h-1} \sum_{i=1}^{2^\ell} 
        			\Pr [ \E_{a,b,\ell,i} ] \\
        &=~ \sum_{\ell=0}^{h-1} \sum_{i=1}^{2^\ell}
        		\Pr [\E_{a,b,\ell,i} \mid 
                   	 a\in \widehat{Q_{\ell,i}}, b\in \widehat{Q_{\ell,i}}]
                \cdot\Pr[a\in\widehat{Q_{\ell,i}}, b\in\widehat{Q_{\ell,i}}] \\
        &\le~ \sum_{\ell=0}^{h-1} \sum_{i=1}^{2^\ell}\left( 
        		6\cdot \frac{A^*}{\area(\widehat{Q_{\ell,i}})}\cdot
                (\area(\widehat{Q_{\ell,i}}))^2 \right) \\
        & =~ 6\cdot A^* \sum_{\ell=0}^{h-1} \sum_{i=1}^{2^\ell}
        		\area(\widehat{Q_{\ell,i}}) \\
        & \le~ 6\cdot A^* \sum_{\ell=0}^{h-1} 1\\
        & =~ 6\cdot A^*\cdot h.
    \end{align*}
    
    At the bottommost level $h$, we can use that
    $\area(Q_{h,i})\le A^*$ for each $i$ to obtain
    \begin{align*}
        \Pr \left[\bigcup_{i=1}^{2^h} [ab \subset Q_{h,i} ] \right] ~&=~
        \sum_{i=1}^{2^h} \Pr \left[ab \subset Q_{h,i}  \right] \\
        & \le~ \sum_{i=1}^{2^h} \Pr \left[a\in  Q_{h,i},  b\in Q_{h,i}\right] \\
        & =~ \sum_{i=1}^{2^h} (\area(Q_{h,i}))^2 \\
        & \le ~ \sum_{i=1}^{2^h}  A^*\cdot (\area(Q_{h,i})) \\
        & = ~ A^*.
    \end{align*}
    We then note that, if $a$ sees $b$, then the event $\E_{a,b,\ell,i}$ 
    occurs for some $\ell<h$ and $i\le 2^\ell$,
    or $a$ and $b$ are in the same polygon $Q_{h,i}$, where $i\le 2^h$.
    Thus 
    \begin{align*}
    	\Pr \left[ {ab} \subset P \right] 
        ~&\le~
        \Pr \left[ \bigcup_{\ell=0}^{h-1}\bigcup_{i=1}^{2^\ell} 
        					[\E_{a,b,\ell,i} ] \right]
        + \Pr \left[\bigcup_{i=1}^{2^h} [ab \subset Q_{h,i} ] \right]\\
        &\le~ 6\cdot A^*\cdot h + A^*\\
        &  =~ A^* + 6\cdot A^*\cdot \log_{3/2} (1/A^*)\\
        &  \le~ A^* + 12\cdot A^*\cdot \log_{2} (1/A^*).
        \qedhere
    \end{align*}
\end{proof}

In the conference version of our paper we proved and used Theorem~\ref{thm:prob2}.
We included the proof here for completeness and archiving purposes. Also, it is a key contribution of this paper to realize that such connection between the probability of being co-visible and the area of the largest convex body could exist. This result has been improved by Balko et al.~\cite{BJVW15}. Using their new result slightly improves the final running time of our algorithms. Thus, we will use in the rest of our paper the following theorem.

\begin{theorem}[Corollary 4 in~\cite{BJVW15}]
\label{thm:BJVW15}
    Let $P$ be an arbitrary unit-area polygon.
    Let $a$ and $b$ be two points chosen uniformly at random in $P$.
    Then
    \[
        \Pr \left[ {ab} \subset P \right] ~\le~ 
        		180\cdot  A^*(P).
    \]
\end{theorem}

\section{Algorithm}
\label{sec:algorithm}

In this section we discuss the eventual algorithm.
The input to the algorithm is a polygon $P$, a parameter $\eps\in (0,1)$,
and a parameter $\delta\in(0,1)$. Without loss of
generality we assume that $P$ has unit area.
The algorithm, called {\sc LargePotato}, 
is summarized in Figure~\ref{fi:algo}. 
In the first part of the section we explain in detail
each step and the notation that is still undefined.
In the second part we analyze the algorithm.

\begin{figure}
  \ovalbox{
  \begin{minipage}{.95\hsize}
   \centering
   \normalsize
    \begin{minipage}{.92\hsize}
       \begin{algorithm}{\sc LargePotato}{
       \qinput{Unit-area polygon $P$, $\eps \in (0,1)$, and  $\delta\in (0,1)$}}
       find a value $A(P)$ such that $A(P)\le A^*(P)\le C_2 \cdot A(P)$;\\
       $r\qlet 60 /A(P)$;\\
       $best\qlet \emptyset$;\\
       \qrepeat~$3\log_2 (1/\delta)$ times\\
       		$R\qlet $ sample $r$ points uniformly at random in $P$;\\
            \qif $G(P,R)$ has at most $C_3 \cdot n$ 
            edges \qthen\\
            	compute $G(P,R)$;\\
                \qfor $ab\in E(G(P,R))$ \qdo\\
                	$R_{ab}\qlet$ sample $96\cdot C_1\cdot C_2/(\eps/2)^{3/2}$ points 
                    	uniformly at random in the parallelogram $\Gamma(a,b,C_2\cdot A(P))$;\\
                    $S_{ab}\qlet$ sample $288\cdot C_2/\eps$ points 
                    	uniformly at random in the parallelogram $\Gamma(a,b,C_2\cdot A(P))$;\\
                    $G_{ab}\qlet G(P, (R_{ab}\cup S_{ab})\cap P)$;\\
                    \qfor $s\in S_{ab}$ \qdo\\
                    	$U\qlet \varphi(G_{ab},s)$;\\
                    	\qif $\area(U) > \area (best)$ \qthen $best \qlet U$;
                    	\qfi
                     \qrof
                \qrof
            \qfi
        \qendrepeat\\
		\qreturn $\conv(best)$;
       \end{algorithm}
       \end{minipage}
     \end{minipage}
}
    \caption{Algorithm. 
    	The constant $C_1$ is from Lemma~\ref{le:convex}.
        The constant $C_2$ is the approximation factor from Hall-Holt et al.~\cite{hkms-06}; see Section~\ref{sec:description}.
        The constant $C_3$ is from Lemma~\ref{le:sizeG}.
    }
    \label{fi:algo}
\end{figure}

\subsection{Description}
\label{sec:description}

\paragraph{Sampling points.}
Let $A(P)$ be a constant-factor approximation for $A^*(P)$. Thus,
$A(P)\le A^*(P) \le C_2 A(P)$ for some constant $C_2\ge 1$.
Hall-Holt et al.~\cite{hkms-06} provide an algorithm to compute such value $A(P)$ in $O(n\log n)$ time.

Let us define $r\mathrel{\mathop:}=\frac{60}{A(P)}$.
Since the largest triangle in any triangulation of $P$ has area at least $1/n$,
we have $A^*(P)\ge 1/n$ and thus $r = O(n)$.

Let $R$ be a sample of $r$ points chosen independently at random from the polygon $P$.
The sample $R$ can be constructed in $O(n + r\log n)$ time, as follows.
By the linear-time algorithm\footnote{Computing a triangulation of $P$ is not the bottleneck of our algorithm. Since Chazelle's algorithm is complicated, for practical purposes it would be easier to use a simpler triangulation algorithm running in $O(n \log n)$ time such as the one described in~\cite{CGbook}.} of Chazelle~\cite{triang-linear-time}, 
we compute a triangulation of $P$, giving triangles $T_1,\dots,T_{n-2}$.
We then compute the prefix sums $S_i=\area(T_1)+\dots+\area(T_i)$
for $i=1,\dots,n-2$. 
This is done in $O(n)$ time.
To sample a point, we select a random number $x$ in the interval $[0,1]$,
perform a binary search to find the smallest index $j$ such that
$x\le S_j$, and sample a random point inside $T_j$. 
A random point inside $T_j$ can be generated using a random
point inside a parallelogram that contains two congruent copies of $T_j$;
such a point can be generated using two random numbers in the interval $[0,1]$.
In total, each point takes $O(\log n)$ time plus the time needed to generate
three random numbers in the interval $[0,1]$.
A similar approach is described in~\cite{shop}.

\paragraph{Size of the visibility graph.}
Using the expected number of edges in the visibility graph $G(P,R)$ 
and Markov's inequality lead to the following bound.

\begin{lemma}
\label{le:sizeG}
    There exists a constant $C_3>0$ such that, 
    with probability at least $5/6$, the graph $G(P,R)$ has
    at most $C_3 \cdot n$ edges.
\end{lemma}
\begin{proof}
    In this proof we use $G\mathrel{\mathop:}=G(P,R)$.
    Using linearity of expectation,
    Theorem~\ref{thm:BJVW15}, the estimates $1/n \le A^*(P)\le C_2A(P)$ and the obvious fact that $n\ge 3$,
    we obtain
    \begin{align*}
    	\EE [ |E(G)|] ~&=~ \binom{r}{2} \cdot 
        	\Pr[\text{two random points are visible in $P$}]\\
            &\le~ \frac 12 \left(\frac{60}{A(P)}\right)^2
            	\cdot 180\cdot  A^*(P)\\
            &\le ~ 324000 \cdot \frac{A^*(P)}{A(P)}\cdot \frac{1}{A(P)}\\
            &\le~ 324000 \cdot C_2 \cdot C_2\cdot n.
    \end{align*}
    Let us take $C_3= 6\cdot 324000 \cdot (C_2)^2$.
    By Markov's inequality we have
\belowdisplayskip=-8pt
    \[
    	\Pr [ |E(G)|\ge C_3\cdot n] ~\le~ 
          \frac{\EE [ |E(G)|]}{C_3\cdot n} ~\le~
          \frac 16\, . 
    \] 
\qedhere
\end{proof}

\paragraph{Constructing the visibility graph and checking its size.}
We will use the following result by Ben-Moshe et al.~\cite{visibility-graph}.

\begin{theorem}[Ben-Moshe et al.~\cite{visibility-graph}]
\label{thm:visibility}
    Let $P$ be a simple polygon with $n$ vertices and 
    let $R$ be a set of $r$ points inside $P$.
    The visibility graph $G(P,R)$ can be constructed in time 
    $O(n + r \log r \log (rn)+k)$,
    where $k$ is the number of edges in $G(P,R)$.
\end{theorem}
In line 6 of the algorithm {\sc LargePotato}, we want to 
check whether $G(P,R)$ has at most $C_3 \cdot n$ edges. 
For this we use that the algorithm of 
Theorem~\ref{thm:visibility} is output-sensitive and takes time 
$T_{\text{\cite{visibility-graph}}}(n,r,k)=O(n + r \log r \log (rn)+k)$.
We run the algorithm of Theorem~\ref{thm:visibility} for at most
$T_{\text{\cite{visibility-graph}}}(n,r,C_3 \cdot n)$ 
steps. If the construction of $G(P,R)$ is not finished, we know that
$|E(G(P,R))|> C_3 \cdot n$. Otherwise the algorithm outputs whether $|E(G(P,R))|\le C_3 \cdot n$ or not.
Thus, the test in line 6 can be made in time
    \begin{align*}
    	T_{\text{\cite{visibility-graph}}}(n,r,C_3 \cdot n)
        ~&=~ O(n + r \log r \log(rn) + C_3 \cdot n) \\
        &=~ O(n + n\log^2 n + n)\\
        &=~ O(n\log^2 n).
     \end{align*}
The construction in line 7 takes the same time, if it is actually made.

\begin{remark} 
For each constant $\eps$, the bottleneck in the running time of our algorithm is here, in our use of Theorem~\ref{thm:visibility} to compute the visibility graph. With the improvement of Balko et al.~\cite{BJVW15}, stated in Theorem~\ref{thm:BJVW15}, all other steps can be made to run in time $O(n\log n \log (1/\delta))$ (for constant $\eps$).
\end{remark}

\paragraph{Work for each edge $ab$.}
We now discuss the work done in lines 9--14 for each edge $ab$ of $G(P,R)$.
The parallelogram $\Gamma(a,b,C_2\cdot A(P))$ was defined in 
Section~\ref{sec:gamma}. Note that $\Gamma(a,b,C_2\cdot A(P))$ has area 
\[
	12\cdot C_2 \cdot A(P) ~\le~ 12\cdot C_2 \cdot A^*(P) ~=~ \Theta (A^*(P)). 
\]
Since $\Gamma(a,b,C_2\cdot A(P))$ is a parallelogram, it is straightforward to
construct the random samples $R_{ab}$ and $S_{ab}$. 
Note that $|R_{ab}|=\Theta(\eps^{-3/2})$ and $|S_{ab}|=\Theta(\eps^{-1})$.
We select the subset of $R_{ab}\cup S_{ab}$ contained in the polygon $P$ 
and construct its visibility graph $G_{ab}$.
We then compute a maximum-area convex clique in $G_{ab}$ among
those cliques whose highest vertex $s$ is from $S_{ab}$. We make this restriction
to reduce the number of candidate highest points
from $\Theta (\eps^{-3/2})$ to $\Theta (\eps^{-1})$.
This is equivalent to computing $\varphi(G_{ab},s)$ for each $s\in S_{ab}$, which is discussed in Section~\ref{sec:phi}.
Finally, we compare the solutions $U_{ab}$ that we obtain
against the solution stored in the variable $best$ and, if appropriate,
update $best$.

\subsection{Analysis}

\begin{lemma}[Time bound]
\label{le:time}
	For each $\eps \in (0,1)$,
    the algorithm {\sc LargePotato} can be adapted to use 
    $O\left( n(\log^2 n + (1/\eps^3) \log n + 1/\eps^4) \log(1/\delta) \right)$ time.
\end{lemma}
\begin{proof}
    The value $A(P)$ can be computed in time $O(n\log n)$, as discussed before.

	We first preprocess the polygon $P$ for segment containment using
    the algorithm of Chazelle et al.~\cite{ray-shooting}:
    after $O(n)$ preprocessing time we can answer whether
    a query segment is contained in $P$ in $O(\log n)$ time.
    In particular, we can decide in $O(\log n)$ time
    whether a query point is in $P$. 
    
    We claim that each iteration of the for-loop (lines 9--14) takes 
    $O((1/\eps^3) \log n + 1/\eps^4)$ time. 
    The samples $R_{ab}$ and $S_{ab}$ can be constructed
    in $O(|R_{ab}|+|S_{ab}|)=O((1/\eps^{3/2}) \log n)$. 
    We construct $(R_{ab}\cup S_{ab})\cap P$
    by testing each point of $R_{ab}\cup S_{ab}$ for containment in $P$.
    The graph $G_{ab}$ is constructed
    by checking for each pair of points from $(R_{ab}\cup S_{ab})\cap P$
    whether the corresponding segment is contained in $P$.
    Thus $G_{ab}$ is constructed in 
    $O((1/\eps^{3/2})^2 \log n)= O((1/\eps^3)\log n)$ time.
    By Lemma~\ref{le:phi}, each iteration of the lines 13--14 takes time 
    $O(|R_{ab}|^2) = O(1/\eps^3)$.
    Thus the running time of the for loop in lines 12--14
    takes time $O(|S_{ab}|\cdot (1/\eps^3))= O(1/\eps^4)$.    
    The claim follows.
    
    We next show that each iteration of the repeat-loop
    (lines 5--14) takes $O(n \log^2 n + (n/\eps^3) \log n + n/\eps^4)$ time. 
    Since $r=O(n)$, the sample $R$ can be computed
    in $O(n \log n)$ time, as discussed in 
    Section~\ref{sec:description}. 
    As discussed before, we can make the test in line 6 in
    $O(n \log^2 n)$ time.
    
    If $G(P,R)$ has more than $C_3 \cdot n$ edges,
    this finishes the time spent in the iteration.
    Otherwise, we make $O(C_3\cdot n)=O(n)$ iterations
    of the for-loop in lines 9--14. Since each iteration 
    of the for-loop takes $O((1/\eps^3) \log n + 1/\eps^4)$ time, 
    as argued earlier in this proof, the bound per iteration 
    of the repeat-loop follows.
\end{proof}

\begin{lemma}[Correctness of one iteration]%
\label{le:correctness}
    In one iteration of the repeat-loop (lines 5--14) of the algorithm
    {\sc LargePotato}
    the algorithm finds a convex polygon of area at least $(1-\eps)A^*(P)$
    with probability at least $1/4$.
\end{lemma}
\begin{proof}
	Let $K^*$ be a convex polygon of largest area contained in $P$.
    Therefore $\area(K^*)=A^*(P)$.     
    Consider one iteration of the repeat-loop. 
    Suppose $G(P,R)$ passes the test on line 6. Then the following two
conditions are sufficient for a successful iteration:
    $R$ contains two visible points $a$ and $b$ such that the parallelogram $\Gamma(a,b,C_2\cdot A(P))$ (used in lines 9--14) contains $K^*$, and $S_{ab}$ contains a point $s$ such that $\area(\varphi(G_{ab},s))$ (computed in line 13) is a $(1-\eps)$-approximation to $\area(K^*)$. This motivates the definition of the following events:
    \begin{align*}
       \E_{K^*}:&~~~ \text{for some edge $ab$ of $G(P,R)$, 
       		$K^*$ is contained in $\Gamma(a,b,C_2\cdot A(P))$},\\
       \E_{G}:&~~~ |E(G(P,R))|~\le~ C_3 \cdot n,\\
       \E_{\Gamma}:&~~~\text{for some edge $ab$ of $G(P,R)$, there is $s\in S_{ab}$   such that}\\ &~~~~~~~~~~~~~~~~~ \area(\varphi(G_{ab},s))~\ge~(1-\eps)\cdot A^*(P).   
    \end{align*}
 	
    Since
    \[
    	|R| ~=~ \frac{60}{A(P)}
        ~\ge~ \frac{60}{\area(K^*)}
    \]
    and $A^*(P)\le C_2\cdot A(P)$,
    Lemma~\ref{le:samplegamma} implies that
    \[
    	\Pr\left[\E_{K^*}\right]\ge \tfrac23.
    \]
    By Lemma~\ref{le:sizeG} we have 
    \[
    	\Pr[\E_{G}] \ge \tfrac 56
    \] 
    and therefore, since $\Pr\left[ A \cap B \right]\geq \Pr\left[ A\right]+\Pr\left[ B\right]-1$,
    \begin{equation}
    \label{eq:pr1}
    	\Pr\left[\E_{K^*} \text{ and } \E_G \right]\ge \tfrac12.
    \end{equation}
    
    For the rest of the proof, we assume that $\E_{K^*}$ and $\E_G$ hold.
    Let $a_0 b_0$ be the edge of $G(P,R)$ such that 
    $\Gamma_0\mathrel{\mathop:}=\Gamma(a_0,b_0,C_2\cdot A(P))$ contains $K^*$.
    The algorithm executes the code in lines 9--14 for $ab=a_0b_0$.
	Let $K^*_{\eps/2}$ be the portion of $K^*$ above $y=y_{1-\eps/2}(K^*)$
    and let $K^*_{1-\eps/2}=K^*\setminus K^*_{\eps/2}$.
    It holds that 
    \[
    	\area(K^*_{\eps/2})=(\eps/2)\cdot \area(K^*) ~\text{ and }~ 
        \area(K^*_{1-\eps/2})=(1-\eps/2) \cdot \area(K^*).
    \]
	The bound
    \[
    	|S_{a_0b_0}| ~=~ \frac{288\cdot C_2}{\eps}
        ~=~ 4\cdot 3 \cdot 
        	\frac{12\cdot C_2 \cdot A^*(P)}{(\eps/2) \cdot A^*(P)}
        ~\ge~ 4\cdot 3 \cdot
        	\frac{\area(\Gamma_0)}{\area(K^*_{\eps/2})}
    \]
    and Lemma~\ref{le:sampleinside} (with $P=\Gamma_0$ and $K=K^*_{\eps/2}$)
    imply that
	\begin{equation}\label{eq:prob1}
    	\Pr\left[ S_{a_0b_0}\cap K^*_{\eps/2} \not= \emptyset
        		\mid \E_{K^*} \text{ and } \E_G \right] 
            ~\ge~ \frac 56\, .
    \end{equation}
    The bound
    \[
    	|R_{a_0b_0}| ~=~ \frac{96\cdot C_1\cdot C_2}{(\eps/2)^{3/2}}
        ~=~ 4\cdot\frac{C_1}{(\eps/2)^{3/2}} \cdot 
        	\frac{12\cdot C_2 \cdot A^*(P)}{A^*(P)/2}
        ~\ge~ 4\cdot \frac{C_1}{(\eps/2)^{3/2}} \cdot
        	\frac{\area(\Gamma_0)}{\area(K^*_{1-\eps/2})}
    \]
    and Lemma~\ref{le:sampleinside2} (with $P=\Gamma_0$ and $K=K^*_{1-\eps/2}$)
    imply that
	\[
    	\Pr\left[ \area(\conv(R_{a_0b_0}\cap K^*_{1-\eps/2}))\ge (1-\eps/2)\cdot \area(K^*_{1-\eps/2}) 	
        	\mid \E_{K^*} \text{ and } \E_G \right] 
            ~\ge~ \frac 23\, .
    \]
	Noting that
    \begin{align*}
    	(1-\eps/2)\cdot \area(K^*_{1-\eps/2})
        	~&=~ (1-\eps/2)\cdot (1-\eps/2)\cdot A^*(P)\\
            &>~ (1-\eps)\cdot A^*(P),
    \end{align*}
    we have
    \begin{equation}\label{eq:prob2}
    	\Pr\left[ \area(\conv(R_{a_0b_0}\cap K^*_{1-\eps/2}))\ge (1-\eps)\cdot A^*(P)	
        	\mid \E_{K^*} \text{ and } \E_G \right] 
            ~\ge~ \frac 23\, .
    \end{equation}
    Joining \eqref{eq:prob1} and \eqref{eq:prob2} we obtain that,
    with probability at least $1/2$, it holds
    \[
    	\area(\conv(R_{a_0b_0}\cap K^*_{1-\eps/2}))\ge (1-\eps)\cdot A^*(P) ~\text{ and }~
        S_{a_0b_0}\cap K^*_{\eps/2} \not= \emptyset.
    \]
    If these two events occur and $s$ is a point of $S_{a_0b_0}\cap K^*_{\eps/2}$, then
    \begin{align*}
    	\area(\varphi(G_{a_0b_0},s)) ~&\ge~
        \area( \conv((K^*_{1-\eps/2}\cap R_{a_0b_0})\cup \{ s \})) \\
        &\ge~ \area( \conv(K^*_{1-\eps/2}\cap R_{a_0b_0}))\\
        &\ge~ (1-\eps)\cdot A^*(P).
    \end{align*}
    We conclude that
	\[
    	\Pr\left[ \E_{\Gamma} \mid \E_{K^*} \text{ and } \E_G \right] 
            ~\ge~ \tfrac 12\, 
    \]
	and using \eqref{eq:pr1} obtain
	\[
    	\Pr\left[ \E_{K^*} \text{ and } \E_G \text{ and } \E_{\Gamma} \right] 
            ~=~ \Pr\left[ \E_{\Gamma} \mid \E_{K^*} \text{ and } \E_G \right]
            	\cdot \Pr\left[ \E_{K^*} \text{ and } \E_G \right] 
            ~\ge~ \frac 12\cdot \frac 12 
            ~=~ \frac 14 \, .
    \]
	When $\E_{K^*}$, $\E_{G}$ and $\E_{\Gamma}$ occur,
    the test in line 6 is satisfied and 
    in one of the iterations of the loop in lines 13--14 we will obtain
    a $(1-\eps)$-approximation to $A^*(P)$.
\end{proof}

\begin{theorem}
\label{thm:main-result}
    Let $P$ be a polygon with $n$ vertices,
    let $\eps$ and $\delta$ be parameters with $0<\eps< 1$ and $0<\delta <1$.
    In time $O\left( n (\log^2 n + (1/\eps^3) \log n + 1/\eps^4) \log(1/\delta) \right)$ 
    we can find a convex polygon contained in $P$ that,
    with probability at least $1-\delta$,
    has area at least $(1-\eps)\cdot A^*(P)$.
\end{theorem}
\begin{proof}
    We consider the output $K$ given by {\sc LargePotato}($P,\eps,\delta$).
    By Lemma~\ref{le:time}, we can assume that the output
    is computed in time $O\left( n (\log^2 n + (1/\eps^3) \log n + 1/\eps^4) \log(1/\delta) \right)$.
    
    The polygon $K$ returned by {\sc LargePotato}($P,\eps,\delta$)
    is always a convex polygon contained in $P$. 
    We have $\area(K)< (1-\eps)\cdot A^*(P)$ if and only
    if all iterations of the repeat-loop (lines 5--14)
    fail to find such a $(1-\eps)$-approximation. Since
    each such iteration fails with probability at most $3/4$
    due to Lemma~\ref{le:correctness}, 
    and there are $3\log_2(1/\delta)$ iterations, we have
\belowdisplayskip=-8pt
    \[
    	\Pr[\area(K)< (1-\eps)\cdot A^*(P)] 
        ~\le~ \left( \frac 34\right)^{3\log_2(1/\delta)}
        ~<~ 
        \left( \frac{1}{2}\right)^{\log_2(1/\delta)} ~=~ \delta.
    \]
\qedhere
\end{proof}

We observe that, if we perform only four iterations of the repeat-loop, we obtain the result stated in the abstract. It is also interesting to note that, when $\eps$ is constant, 
then the running time of the algorithm is $O(n \log^2 n \log(1/\delta) )$.

\section{Convex body of maximum perimeter}
\label{sec:perimeter}

In this section we present an adaptation of the previous algorithm in order to maximize the perimeter. 
Recall that for a convex body $K$ in the plane, 
we denote its perimeter by $\per(K)$, and we have defined
\begin{align*}
    L^*(P) ~=~ \max \{ \per(K) \mid K\subset P,~ K \text{ convex}\}.
\end{align*}
Before presenting the actual algorithm, we first introduce a few tools.

\subsection{Perimeter of the convex hull of random samples}
Let $K$ be a convex body in the plane and let $R$ be a random sample
of $n$ points inside $K$. How well does $\per(\conv(R))$ approximate $\per(K)$?
There has been quite some research on this problem, often on the 
high-dimensional generalizations called intrinsic volumes. 
See~\cite{Schneider2004,Schneider2008} for an overview of the known results on this problem.
However, all the results we found use constants that depend on $K$.
Since in our application the target convex body $K$ is unknown, we are
interested in bounds using universal constants, independent of $K$.

\begin{theorem}
\label{thm:randomperimeter}
	There is some universal constant $C_4\ge 1$ such that the following holds.
	For each convex body $K$ in the plane, if 
	$K_m$ denotes the convex hull of $m$ points chosen uniformly
	at random inside $K$, then
	$\per(K)-\EE [\per(K_m)] \le \tfrac{C_4}{\sqrt{m}}\cdot \per (K)$.
\end{theorem}
\begin{proof}	
	Let us first provide some notation used only in this proof.
	In the notation we drop the dependency on $K$.
    Let $\U$ be the set of unit vectors in the plane.    
	For each $u\in \U$ and each $t\in \RR_{\ge 0}$ we define the following values;
	see Figure~\ref{fig:randomperimeter}.
	\begin{align*}
		h(u) ~&\mathrel{\mathop:}=~ \max \{\scalar{p,u} \mid p\in K \},\\
		\dwidth(u) ~&\mathrel{\mathop:}=~ h(u)+h(-u),\\
		S(u,t) ~&\mathrel{\mathop:}=~ \{ p \in \RR^2 \mid h(u)-t\le\scalar{p,u} \le h(u)\}, \\
		\ell(u,t) ~&\mathrel{\mathop:}=~ \{ p \in \RR^2 \mid h(u)-t = \scalar{p,u}\}, \\
		v(u,t) ~&\mathrel{\mathop:}=~ \area( K\cap S(u,t)).
	\end{align*}
	Note that $\dwidth(u)$ is the so-called \emph{directional width} of $K$ in direction $u$:
	the length of the orthogonal projection of $K$ onto any line parallel to $u$.
	The line $\ell(u,t)$ is perpendicular to $u$ and $\ell(u,0)$ is tangent to $K$.
	Moreover, $S(u,t)$ is an infinite slab of width $t$ defined by $\ell(u,0)$ and $\ell(u,t)$.
	When $t>0$ the slab $S(u,t)$ intersects the interior of $K$. The value
	$v(u,t)$ tells the area of the portion of $K$ contained in $S(u,t)$.

	\begin{figure}
		\includegraphics[width=\textwidth]{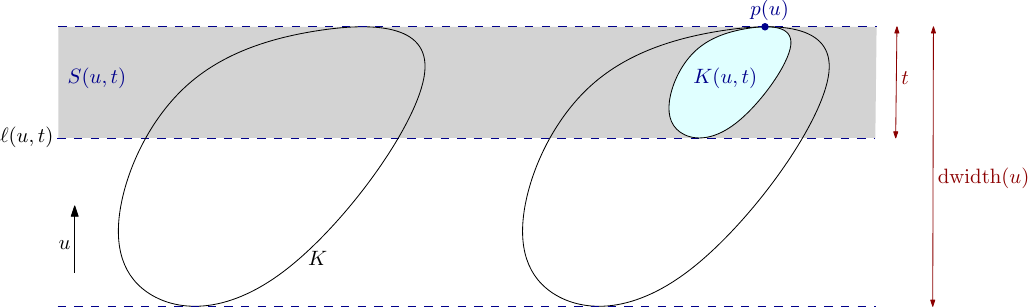}
		\caption{Notation in the proof of Theorem~\ref{thm:randomperimeter}.}
		\label{fig:randomperimeter}
	\end{figure}
	
	In the proof we are going to use the classical Crofton's formula that tells
	\begin{equation}
		\per(K) ~=~ \int_\U \dwidth(u) \, d\omega(u) 
				~=~ \int_\U \int_{-\infty}^{+\infty} 
						\mathds{1}_{K\cap \ell(u,t)\not= \emptyset} \, dt \, d\omega(u),
	\label{eq:Crofton}
	\end{equation}
	where $\omega$ is the uniform Lebesgue measure on $\U$ satisfying $\int_\U 1 \, d\omega(u)=1/2$.
	
	We adapt the approach of Schneider and Wieacker~\cite{S87,SW80},
	based on an observation of Efron~\cite{Efron65}.
	
	By scaling, we can assume that $K$ has unit area. 
	Consider a line $\ell(u,t)$ that intersects $K$. 
	The line $\ell(u,t)$ does not intersect $K_m$ if and only if all points of the random sample
	lie in the interior of $S(u,t)$ or all the points of the random sample lie in $K\setminus S(u,t)$.
	Since $K$ has area $1$ and $K\cap S(u,t)$ has area $v(u,t)$, this means that 
	\[
		\forall u\in \U \text{ and } t\in [0,\dwidth(u)]: ~~~ 
			\Pr[\ell(u,t)\cap K_m \not= \emptyset] ~=~ 1- v(u,t)^m - \bigl(1-v(u,t)\bigr)^m.
	\]
	
	Using Fubini's theorem we have 
	\begin{align*}
		\EE [\per(K_m)] ~&=~ \int_\U \int_0^{\dwidth(u)} \Pr[\ell(u,t)\cap K_m \not= \emptyset] 
			\, dt\, d\omega(u)\\
			&=~ \int_\U \int_0^{\dwidth(u)}  \Bigl[1- v(u,t)^m - \bigl(1-v(u,t)\bigr)^m \Bigr]
			\, dt\, d\omega(u)\\
			&=~ \int_\U \dwidth(u) \, d\omega(u) -
				\int_\U \int_0^{\dwidth(u)} \Bigl[ v(u,t)^m + \bigl(1-v(u,t)\bigr)^m \Bigr]
					\, dt\, d\omega(u)\\
			&=~ \per(K) - \int_\U \int_0^{\dwidth(u)} \Bigl[ v(u,t)^m + \bigl(1-v(u,t)\bigr)^m \Bigr]
					\, dt\, d\omega(u).
	\end{align*}
	Note that for each $u\in \U$ and each $t$ with $0\le t \le \dwidth(u)$ we have
	\[
		1-v(u,t)= v(-u,\dwidth(u)-t),
	\]
	and therefore, by applying the change of variables $w=-u$, $s=\dwidth(u)-t$ and renaming the new variables as $u$, $t$,
	\begin{align*}
		\int_\U \int_0^{\dwidth(u)} \bigl(1-v(u,t)\bigr)^m dt\, d\omega(u)
		~&=~ \int_\U \int_0^{\dwidth(u)} v(u,t)^m dt\, d\omega(u).
	\end{align*}
	Thus, rearranging terms we get
	\begin{align*}
		\per(K) - \EE[\per(K_m)] ~&=~ 
			2\, \int_\U \int_0^{\dwidth(u)} \bigl(1-v(u,t)\bigr)^m dt\, d\omega(u).
	\end{align*}
	
	Now we adapt the estimate of Schneider~\cite{S87} to make it independent of $K$.
	The right of Figure~\ref{fig:randomperimeter} may be helpful.
	For each $u\in \U$, let $p(u) \in K$ be a point maximizing $\scalar{p,u}$.
	For each $u\in \U$ and each $t\in \RR_{\ge 0}$,
	let $K(u,t)$ be a copy of $K$ scaled by $\frac{t}{\dwidth(u)}$ with center $p(u)$.
	Note that $K(u,t)$ is contained in $S(u,t)$ and, if $0\le t\le \dwidth(u)$, it is also contained in $K$. 
	Therefore, since $\area(K)=1$, we have
	\[
		\forall u\in \U \text{ and } t\in [0,\dwidth(u)]: ~~~ 
			v(u,t) \ge \area(K(u,t)) = \left(\frac{t}{\dwidth(u)}\right)^2.
	\]
	Thus we have the estimate 
	\begin{align*}
		\per(K) - \EE[\per(K_m)] ~&\le~ 
			2\, \int_\U \int_0^{\dwidth(u)} \left(
					1-\left(\frac{t}{\dwidth(u)}\right)^2\right)^m dt\, d\omega(u).
	\end{align*}
	For each $u\in \U$ we can use the change of variable $x=(t/\dwidth(u))^2$ and we obtain
	that
	\begin{align*}
		\per(K) - \EE[\per(K_m)] ~&\le~ 
			\int_\U \int_0^1 x^{-1/2} (1-x)^m \dwidth(u)\, dx\, d\omega(u)\\
			&=~ \int_\U \dwidth(u) \int_0^1 x^{-1/2} (1-x)^m\, dx \, d\omega(u).
	\end{align*}
	Using the standard formula 
	\[
		\int_0^1 x^{a-1} (1-x)^{b-1} \,dx ~=~ \frac{\Gamma(a)\Gamma(b)}{\Gamma(a+b)}
	\]
	for the beta function and the gamma function, and the known value $\Gamma(1/2)=\sqrt{\pi}$,
	we further derive
	\begin{align*}
		\per(K) - \EE[\per(K_m)] ~&\le~ 
			\int_\U \dwidth(u) \int_0^1 x^{-1/2} (1-x)^m\, dx \, d\omega(u)\\
			&=~ \int_\U \dwidth(u) \frac{\sqrt{\pi} \,\Gamma(m+1)}{\Gamma(m+3/2)}\, d\omega(u)\\		
			&=~ \frac{\sqrt{\pi} \,\Gamma(m+1)}{\Gamma(m+3/2)} \int_\U \dwidth(u)\, d\omega(u)\\		
			&=~ \frac{\sqrt{\pi} \,\Gamma(m+1)}{\Gamma(m+3/2)} \cdot \per(K).		
	\end{align*}
	Now we use that
	\[
		\lim_{m\rightarrow \infty} \, \frac{\Gamma(m+3/2)}{\Gamma(m+1)\, \sqrt{m}}=1
	\]
	to conclude that, for some constant $C_4\ge 1$,
	\[
		\per(K) - \EE[\per(K_m)] ~\le~ \frac{C_4}{\sqrt{m}}\,\per(K). \qedhere
	\]
\end{proof}

Note that the bound in this theorem is optimal. 
When $K$ is an equilateral triangle of unit area, to get a $(1-\eps)$-approximation of $\per(K)$,
we need to sample at least one point at distance at most $O(\eps)$ from each vertex, 
and these regions have area $\Theta(\eps^2)$.

The following lemma is the analogue to Lemma~\ref{le:sampleinside2} for the perimeter.

\begin{lemma}\label{le:sampleinside3}
	Let $K$ be a convex body contained in a polygon $P$, 
	let $R$ be a random sample of points inside $P$,
    and let $C_4$ be the constant in Theorem~\ref{thm:randomperimeter}.
    If 
    \[
    	|R|~\ge~ 4\cdot \left( 6 C_4/\eps\right)^2 \cdot \frac{\area(P)}{\area(K)}\,,
    \] 
    then with probability at least $2/3$ it holds that
    $\per(\conv(R\cap K))\ge (1-\eps)\per(K)$.
\end{lemma}
\begin{proof}
	We define the following events:
    \begin{align*}
       \E:&~~~ |R\cap K|~\ge~ (6C_4/\eps)^{2},\\
       \F:&~~~ \per(\conv(R\cap K))~\ge~ (1-\eps)\cdot \per(K) .   
    \end{align*}
    For each event $\mathcal{A}$ we use $\overline{\mathcal{A}}$ 
    for its negation.
    Since $C_4\ge 1$, then $(6C_4/\eps)^{2}\ge 3$ and 
    Lemma~\ref{le:sampleinside} implies
    \[
    	\Pr\left[ \E \right] ~\ge~ \frac 56\, .
    \]
    Assuming the event $\E$, that is, $|R\cap K|~\ge~ (6C_4/\eps)^{2}$,
	it follows from Markov's inequality and Theorem~\ref{thm:randomperimeter} that
	\begin{align*}
		\Pr\left[\per(K)-\per(\conv(R\cap K)) \ge \eps\cdot \per(K)\right] ~&\le~ 
			\frac{1}{\eps\cdot \per(K)}\cdot\EE\left[ \per(K)-\per(\conv(R\cap K)) \right]\\
			&\le~ \frac{1}{\eps\cdot \per(K)}\cdot\frac{C_4}{\sqrt{|R\cap K|}}\cdot\per(K)\\
			&\le~ \frac{C_4}{\eps\cdot \per(K) \cdot(6C_4/\eps)} \cdot \per(K)\\
			&=~ \frac 16\, .
	\end{align*}
	This means that 
    \[
    	\Pr\left[ \overline{\F} \mid \E \right] ~\le~ \frac 16\, 
    \]
    and therefore 
    \[
		\Pr\left[ \overline{\F}\right] ~ \le~ \Pr\left[ \overline{\F} \mid \E \right] + 
        		\Pr\left[ \overline{\E}\right] ~\le~ \frac 16 +\frac 16 ~=~ \frac 13\, .
	\qedhere
    \]
\end{proof}

\subsection{Bounds depending on the fatness}
We recall a result about approximation of convex bodies in the plane by rectangles: 

\begin{lemma}[Schwarzkopf et al.~\cite{Schwarzkopf-rect}]
Given a convex body $K$ in the plane, there exist two similar and parallel rectangles $\piin{K}$ and $\piout{K}$ such that $\piin{K}\subseteq K \subseteq \piout{K}$ and the sides of $\piout{K}$ are at most twice as long as the sides of $\piin{K}$.
\end{lemma}

Note that the statement does \emph{not} guarantee that 
$\piin{K}$ and $\piout{K}$ have a common center.
Given a convex body $K$, we denote by $d_1(K)$ and $d_2(K)$ the lengths of the sides of $\piin{K}$, with $d_1(K) \geq d_2(K)$. Although the rectangles $\piin{K}$ and $\piout{K}$ are not necessarily unique, 
this does not affect our arguments.

We will run two algorithms and choose the best between both outputs.
One of the algorithms is a $(1-26\eps)$-approximation when the optimal solution $K^*$ 
satisfies $d_2(K^*)/d_1(K^*)\le \eps$, while the other covers the case $d_2(K^*)/d_1(K^*)\ge \eps$.
We now develop bounds for both cases.

\begin{lemma} \label{le:stick}
Let $K^*$ be a convex body contained in $P$ such that $\per(K^*)= L^*(P)$.
Let $\ell$ be the length of a longest line segment contained in $P$.
If $d_2(K^*)/d_1(K^*)\le \eps \le 2/5$,
then $L^*(P) \leq 2\cdot \ell \cdot (1+25\cdot \eps)$.
\end{lemma}

\begin{proof}
	To simplify the notation, 
	in this proof we set $d_1:=d_1(K^*)$ and 
	$d_2:=d_2(K^*)$. Without loss of generality, 
	we assume that the longer side of $\piin{K^*}$ is horizontal, 
	and the shorter one is vertical.

	Let $s$ be a longest line segment contained in $K^*$, 
	let $\ell_{K^*}$ be its length, and let $a$ and $b$ be its two endpoints. 
	Clearly, $\ell_{K^*}\leq \ell$. Since $\piin{K^*}\subset K^*$, 
	we have $\ell_{K^*}\geq d_1$. 

	We first observe that $s$ is not vertical. Indeed, in this case the containment
	$K^*\subset \piout{K^*}$ implies $\ell_{K^*}\le 2\cdot d_2$. 
	Thus, we would have $d_1\le 2\cdot d_2$, which contradicts 
	the assumption that $d_2/d_1\le \eps\le 2/5$.

	Without loss of generality, we assume that $s$ has non-negative slope $\alpha$. 
	Since $s$ is a longest line segment in $K^*$, 
	$K^*$ is contained in the infinite region of the plane bounded 
	by the lines through $a$ and $b$ perpendicular to $s$. 
	We denote by $\Psi$ the parallelogram resulting from 
	the intersection of this region and the region of the plane bounded 
	by the lines supporting the horizontal edges of $\piout{K^*}$.
	See Figure~\ref{fig:skinny-polyg}. 
	The horizontal sides of $\Psi$ have length $\ell_{K^*}/\cos \alpha$, 
	while the other sides have length at most $2\cdot d_2/\cos \alpha$. 
	It is well-known, and an easy consequence of Crofton's formula 
    (see equation~\eqref{eq:Crofton} in Theorem~\ref{thm:randomperimeter}), 
    that if a convex body $K_1$ is contained in a convex body $K_2$, 
	then $\per(K_1)\leq \per(K_2)$. 
	Since $K^* \subseteq\Psi$, it is enough to prove that 
	$\per(\Psi)\leq 2\cdot \ell \cdot (1+25\cdot \eps)$.

	\begin{figure}
		\centering
		\includegraphics[]{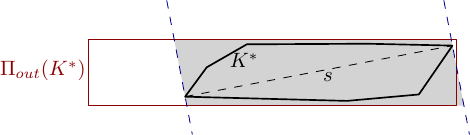}
		\caption{Situation in the proof of Lemma~\ref{le:stick}. In gray, the parallelogram $\Psi$.}
		\label{fig:skinny-polyg}
	\end{figure}

	Since $s$ is contained in $\piout{K^*}$, 
	the maximum slope is attained when one of the endpoints of $s$ 
	lies in the lower side of $\piout{K^*}$, 
	and the other endpoint in the upper side. 
	Therefore, 
	\[
		\sin \alpha ~\le~ \frac{2\cdot d_2}{\ell_{K^*}} ~\le~
			\frac{2\cdot \eps \cdot d_1}{d_1} ~=~ 2\cdot \eps\,.
	\]
	Since $\eps< 1/2$,
	\[
	\cos \alpha ~\ge~ \sqrt{1-4\cdot \eps^2} ~>~ 1-2\cdot \eps\,.
	\]
	Thus, we have
	\begin{align*}
			\per(\Psi) ~&\le~ 
			\frac{2\cdot \ell_{K^*}}{\cos \alpha}+\frac{4\cdot d_2}{\cos \alpha}\\
			&\le~ \frac{1}{\cos \alpha}\cdot (2\cdot \ell+4\cdot \eps \cdot d_1)\\
			&<~ \frac{1}{1-2\cdot \eps}\cdot (2\cdot \ell+4\cdot \eps \cdot \ell)\\
			&=~
			\frac{2\cdot \ell \cdot (1+2\cdot \eps)}{1-2\cdot \eps}\\
			&<~ 2\cdot \ell \cdot (1+25\cdot \eps)\, ,
		\end{align*}
	   where the last inequality holds because $\eps \le 2/5$.
\end{proof}

\begin{lemma}
\label{le:fat}
Let $K^*$ be a convex body contained in $P$ such that $\per(K^*)= L^*(P)$.
If $d_2(K^*)/d_1(K^*)\ge \eps$, then $\area(K^*)\geq \frac{\eps}{16}\cdot A^*(P)$. 
\end{lemma}

\begin{proof}
	To simplify the notation, in this proof we set $d_1:=d_1(K^*)$ and $d_2:=d_2(K^*)$.
	Let $K'$ be a convex shape contained in $P$ with $\area(K')=A^*(P)$. 
	We further define $d'_1:=d_1(K')$ and $d'_2:=d_2(K')$. 
	We have the following obvious relations:
    \begin{align*}
    	\per(K^*) ~&\le~ \per(\piout{K^*} ~\le~ 4\cdot (d_1+d_2)\, ,\\
        \per(K') ~&\ge~ \per(\piin{K'}) ~=~ 2\cdot (d'_1+d'_2)\, , \\
        \area(K^*) ~&\ge~ \area(\piin{K^*}) ~=~   d_1\cdot d_2\, ,\\
        \area(K') ~&\le~ \area(\piout{K'}) ~\le~  4 \cdot d'_1\cdot d'_2\, .
    \end{align*}
	Combining the first two inequalities with $\per(K')\le \per(K^*)$, we obtain 
	\begin{equation}
			d'_1+d'_2\le  2\cdot (d_1+d_2)\, .
	\label{eq:perim1}
	\end{equation}
	
	If $2\cdot d_1\ge d'_1$, then
	\begin{align*}
    	\area(K^*) ~&\ge~ d_1\cdot d_2 ~\ge~ d_1\cdot \eps \cdot d_1
         ~\ge~ (\eps/4)\cdot (d'_1)^2 \\ ~&\ge~ (\eps/4)\cdot d'_1 \cdot d'_2 ~\ge~ 
		(\eps/16)\cdot \area(K')\\
		&=~ (\eps/16)\cdot A^*(P)\,.
    \end{align*}
    
	Let us now consider the case $2\cdot d_1< d'_1$. Since the inequality \eqref{eq:perim1} implies
	\[
		2\cdot (d_1+d_2)~\ge~ d'_1+d'_2 ~>~ 2\cdot d_1 +d'_2\,
	\]
	we obtain $2\cdot d_2> d'_2$. 
	Adding the inequalities
	$d'_1> 2\cdot d_1$ and $2\cdot d_2> d'_2$
	we get that 
	\begin{equation}
			d'_1-d'_2> 2\cdot (d_1-d_2)\, .
	\label{eq:perim2}
	\end{equation}
	Combining \eqref{eq:perim1} and \eqref{eq:perim2} we have
	\begin{align*}
			d_1\cdot d_2 ~&=~ 1/4 \cdot ((d_1+ d_2)^2 - (d_1- d_2)^2)\\
			~&>~ 1/4 \cdot ( 1/4 \cdot (d'_1+ d'_2)^2 - 1/4 \cdot (d'_1- d'_2)^2)\\
			 ~&=~ (1/4)\cdot d'_1\cdot d'_2 \, .
		\end{align*}
	Thus, also in this case we get
	\begin{align*}
			\area(K^*) ~&\ge~ d_1\cdot d_2 ~>~  (1/4) \cdot d'_1\cdot d'_2
		  ~\ge~  (\eps/16)\cdot 4\cdot d'_1 \cdot d'_2 \\~&\ge~  (\eps/16)\cdot \area(K')
			~=~ (\eps/16)\cdot A^*(P) \, .
	\qedhere
	\end{align*}
\end{proof} 

\subsection{Algorithm}

As already mentioned, we run two algorithms to find a $(1-\eps)$-approximation of $L^*(P)$. 
In fact, to keep the computations slightly simpler, we will provide for a $(1-26 \eps)$-approximation of $L^*(P)$. 

Let $K^*$ be a convex body inside $P$ with $\per(K^*)=L^*(P)$.
The first algorithm finds a $(1-26\eps)$-approximation of the value $L^*(P)$ 
when $d_2(K^*)/d_1(K^*) \leq \eps$, 
while the second algorithm returns a $(1-\eps)$-approximation of the value $L^*(P)$
when $d_2(K^*)/d_1(K^*) \ge \eps$. Since both algorithms compute a convex polygon contained in $P$, taking
the best of the two solutions we obtain a $(1-26\eps)$-approximation in any case.
We can assume that $\eps\le 2/5$, as otherwise we can just take $\eps=2/5$.

Consider first the case $d_2(K^*)/d_1(K^*) \leq \eps$. 
This means that the optimal solution $K^*$ is ``skinny".
Let $\ell$ be the length of a longest segment contained in $P$.
Let $\bar{s}$ be a line segment contained in $P$ of length at 
least $(1-\eps)\cdot \ell$. Lemma~\ref{le:stick} implies that
\[
	\per(\bar{s}) ~\ge~
	2\cdot (1-\eps)\cdot \ell ~\ge~ 
	L^*(P)\cdot \frac{1-\eps}{1+25\cdot \eps}
	~\ge~ L^*(P)\cdot (1-26\cdot \eps)  \,.
\]
Hall-Holt et al.~\cite{hkms-06} show how to compute such a segment $\bar{s}$ in $O((n/\eps^4) \log^2 n)$ time. We conclude that, whenever $d_2(K^*)/d_1(K^*) \leq \eps$, we can obtain a $(1-26\cdot \eps)$-approximation to $L^*(P)$ in $O((n/\eps^4) \log^2 n)$ time.

Consider now the case $d_2(K^*)/d_1(K^*) \geq \eps$.
This means that the optimal solution $K^*$ is ``slightly fat".
By Lemma~\ref{le:fat} we have
\[
	\frac{\eps}{16}\cdot A^*(P) ~\le~ \area(K^*) ~\le~ A^*(P).
\]
Let $A(P)$ be the approximation computed in 
the algorithm {\sc LargePotato}, line 1 of Figure~\ref{fi:algo}.
We then know that
\[
	\frac{\eps}{16}\cdot A(P) ~\le~ \area(K^*) ~\le~ C_2\cdot A(P).
\]
We divide the interval $[\frac{\eps}{16}\cdot A(P), C_2\cdot A(P)]$ into the following $O(\log 1/\eps)$ subintervals:
\[
	I_i\mathrel{\mathop:}=[C_2\cdot A(P)/2^{i+1},\, C_2\cdot A(P)/2^i],~~~
	i\in \ZZ, ~ 0\le i \le \lceil\log_2 (16C_2/\eps)\rceil.
\]
For each integer $i$ we can apply a modification of the algorithm {\sc LargePotato}, as follows.

\begin{lemma}
\label{le:iteration}
A modification of the algorithm {\sc LargePotato} taking as an extra parameter an integer $i$ has the following properties.
It always takes time 
\[
O\left( n\left[ 2^i \log^2 n + (4^i/\eps^4) \log n + 4^i/\eps^6\right] \log(1/\delta)\right).\]
If $\area(K^*)$ is in $I_i$, then the algorithm
finds a convex polygon of perimeter at least $(1-\eps)L^*(P)$ with probability at least $1-\delta$.
If $\area(K^*)$ is not in $I_i$, then the algorithm returns a convex polygon inside $P$.
\end{lemma} 
\begin{proof}
As stated in the algorithm {\sc LargePotato}, we assume that $P$ has unit area.
To simplify the computation, set $A_i=C_2\cdot A(P)/2^{i+1}$, that is, the lower endpoint of the interval $I_i$.
The value $r$ in line 2 of the algorithm {\sc LargePotato} is set to $r= 60/A_i$. Since $A(P)=\Omega(1/n)$ we have $r=O(2^i n)$.

Consider one of the iterations of the repeat loop (lines 5--14).
A slight modification of the proof of Lemma~\ref{le:sizeG} gives that the
expected size of the visibility graph $G(P,R)$ is bounded by
\begin{align*}
    	\EE [ |E(G(P,R))|] ~&=~ \binom{r}{2} \cdot 
        	\Pr[\text{two random points are visible in $P$}]\\
            &\le~ \frac 12 \left(\frac{60}{A_i}\right)^2
            	\cdot 180\cdot  A^*(P)\\
            &\le~ 324000\cdot \left(\frac{1}{C_2\cdot A(P)/2^{i+1}}\right)^2
            	\cdot A^*(P) \\
            &\le~ 4^{i+1} \cdot 324000 \cdot n.
\end{align*}
Therefore, the condition in line~6 of {\sc LargePotato} becomes ``\textbf{if} $G(P,R)$ has at most $C' \cdot n$ edges \textbf{then}'',
where $C'=6\cdot 4^{i+1} \cdot 324000$. 
This condition is satisfied in each iteration of the repeat loop with probability at least $5/6$.
This condition can be checked using Theorem~\ref{thm:visibility} in 
\[
O(n+r \log r \log(rn)+ 4^i n)=
O(2^i n \log(2^i n)\log (2^i n^2) + 4^i n)= O(n [4^i + 2^i \log^2 n])
\]
time.

Under the assumption that $A_i\le \area(K^*)\le 2A_i$, 
we have $r\ge 60/\area(K^*)$ and 
Lemma~\ref{le:samplegamma} ensures that, with probability at least $2/3$, 
the body $K^*$ lies in some parallelogram $\Gamma(a_0,b_0,2A_i)$ 
for some edge $a_0b_0$ of the visibility graph $G(P,R)$.

It remains to discuss how to find a maximum-perimeter polygon contained in 
the parallelogram $\Gamma(a_0,b_0,2A_i)$ that contains $K^*$ (lines 9--14).
For the other parallelograms $\Gamma(a,b,2A_i)$, 
$ab\in E(G(P,R))$, we just need to make sure that we find some convex polygon contained in $P$.

Consider an edge $ab$ of the visibility graph $G(P,R)$ and note
that 
\[
	\area(\Gamma(a,b,2A_i))=12\cdot 2A_i=24\cdot A_i.
\]
We have 
\[
	4\cdot (6C_4/\eps)^2 \cdot \frac{\area(\Gamma(a,b,2A_i))}{\area(K^*)}
	~\le~ \frac{144 \cdot (C_4)^2}{\eps^2} 
	\cdot \frac{24\cdot A_i}{A_i} ~=~ 
	\frac{C_5}{\eps^2},
\]
for some constant $C_5$. 

Using Lemma~\ref{le:sampleinside3} we obtain the following: if we take a sample $R_{a_0b_0}$ of $C_5/\eps^2$ points inside $\Gamma(a_0,b_0,2A_i)$, then with probability at least $2/3$, we have $\per(\conv(R_{a_0b_0}\cap K^*))\ge (1-\eps)\cdot L^*(P)$.
Thus, we proceed, for each $ab\in E(G(P,R))$ as follows:
take a sample $R_{ab}$ of $C_5/\eps^2$ points inside $\Gamma(a,b,2A_i)$,
build the visibility graph $G_{ab}$ of $R_{ab}$, and find a convex clique
in $G_{ab}$ of largest perimeter. As discussed in Lemma~\ref{le:time}, the visibility graph can be built in $O(|R_{ab}|^2 \log n)= O((1/\eps)^4 \log n)$ time. For computing the convex clique of largest perimeter, we use the modification of Lemma~\ref{le:phi} mentioned thereafter, 
using each point of $R_{ab}$ as highest point. 
Unlike in the case of approximating the area, here we cannot afford to use an asymptotically smaller sample $S_{ab}$ for the highest points, for the following reason. Let $T$ be an equilateral triangle with a horizontal base at the bottom and a vertex $v$ on top. To approximate the perimeter of $T$ by a convex hull of a set of points inside $T$ with error at most $\varepsilon$, the highest point of the set must be in a region of area $\Omega(\varepsilon^2 \cdot \area(T))$ near $v$. Hence we would need at least $\Omega(1/\varepsilon^2)$ points in the sample $S_{ab}$.
Thus, we need $O(|R_{ab}| \cdot |R_{ab}|^2)= O(1/\eps^6)$ time to compute the convex clique of largest perimeter. We conclude that for each edge $ab$ of $G(P,R)$ we spend $O((1/\eps^4) \log n + 1/\eps^6)
$ time.

Since we make $|E(G(P,R))|=O(4^i n)$ iterations 
of the for loop (lines 9--14), and each iteration of the 
for loop (lines 9--14) takes $O((1/\eps^4) \log n + 1/\eps^6)$ time,
in the for loop of lines 8--14 we spend 
$O(n\cdot 4^i \cdot ((1/\eps^4) \log n + 1/\eps^6))$ time.
To make the test in line $6$ we spend 
$O(n [4^i + 2^i \log^2 n])$. It follows that in each
iteration of the repeat loop we spend
$O\left( n\left[ 2^i \log^2 n + (4^i/\eps^4) \log n + 4^i/\eps^6\right]\right)$ time.
Since the algorithm makes $O(\log (1/\delta))$ iterations of the repeat loop,
the claimed time bound follows.

Under the assumption that $\area(K^*)$ lies in the interval $I_i$, 
the graph $G(P,R)$ passes the test of line 6 with probability at least $5/6$,
one of the parallelograms $\Gamma(a_0,b_0,2A_i)$ 
contains $K^*$ with probability at least $2/3$,
and the sample $R_{a_0b_0}$ has the property that 
$\per(K^*\cap R_{a_0b_0})\ge (1-\eps)L^*(P)$ with probability at least $2/3$.
When all three events occur, the algorithm finds a $(1-\eps)$-approximation. 
As shown in the proof of Lemma~\ref{le:correctness}, with probability at least $1/4$, the three events occur simultaneously, and thus some iteration of the repeat loop 
is successful with probability at least $1-\delta$, 
as shown in the proof of Theorem~\ref{thm:main-result}.
We conclude that, when $\area(K^*)$ lies in the interval $I_i$,
the output of the algorithm is a $(1-\eps)$-approximation with probability at least $1-\delta$.

When $\area(K^*)$ does not lie in the interval $I_i$, we spend the same time and we return a convex polygon contained in $P$ (possibly degenerated to a single point) without any guarantee.
\end{proof}

\begin{lemma}
\label{le:alliterations}
When $d_2(K^*)/d_1(K^*) \geq \eps$, we can find a 
convex polygon of perimeter at least $(1-\eps)L^*(P)$ with probability at least $1-\delta$ in time 
\[
O\left( n \left[ (1/\eps) \log^2 n + (1/\eps^6) \log n + 1/\eps^8\right]  \log(1/\delta)\right).
\]
\end{lemma}
\begin{proof}
We use the algorithm of Lemma~\ref{le:iteration} for each interval $I_i$, 
where $i=0,1,\dots,\allowbreak \lceil\log_2 (16C_2/\eps)\rceil$, and return
the polygon with largest perimeter we get over all iterations.
The running time is
\[
	\sum_{i=0}^{\lceil \log_2 (16C_2/\eps)\rceil} O\left( n \left[2^i \log^2 n + \frac{4^i}{\eps^4} \log n + \frac{4^i}{\eps^6} \right] \log(1/\delta)\right).
\]
Using that $\sum_i 2^i = O(1/\eps)$ and $\sum_i 4^i = O(1/\eps^2)$,
this becomes
\[
O\left( n \left[\frac{1}{\eps} \log^2 n + \frac{1}{\eps^6} \log n + \frac{1}{\eps^8} \right] \log(1/\delta)\right).
\]

The algorithm is successful in getting a $(1-\eps)$-approximation whenever the 
iteration with $\area(K^*)$ in $I_i$ is successful. Thus, the whole algorithm
is a $(1-\eps)$-approximation with probability at least $1-\delta$, for the case 
$d_2(K^*)/d_1(K^*) \geq \eps$.
\end{proof}

Combining the algorithms for $d_2(K^*)/d_1(K^*) \leq \eps$ and $d_2(K^*)/d_1(K^*) \geq \eps$ we get a $(1-26\eps)$-approximation. Replacing $\eps$ with $\eps/26$ in the whole discussion, we obtain the following final result for maximizing the perimeter.

\begin{theorem}
\label{thm:main-result2}
    Let $P$ be a polygon with $n$ vertices,
    let $\eps$ and $\delta$ be parameters with $0<\eps< 1$ and $0<\delta <1$.
    In time $O\left( n\left[ (1/\eps^4) \log^2 n + \left((1/\eps) \log^2 n + (1/\eps^6) \log n + 1/\eps^8\right) \log(1/\delta) \right]  \right)$ 
    we can find a convex polygon contained in $P$ that,
    with probability at least $1-\delta$,
    has perimeter at least $(1-\eps)\cdot L^*(P)$.
\end{theorem}

\section{Conclusions}
\label{sec:conclusions}
There are several directions for future work. 
We explicitly mention the following:
\begin{itemize}
	\item Finding a deterministic $(1-\eps)$-approximation 
    	using near-linear time.
    \item Achieving subquadratic time for polygons with an unbounded number
    	of holes.
\end{itemize}

In the conference version of this paper (in the proceedings of SoCG 2014), we also mentioned the following two questions that have been answered affirmatively
by Balko et al.~\cite{BJVW15}.

\begin{itemize}
    \item Does Theorem~\ref{thm:prob1}(i) hold for arbitrary simple polygons?
    	We conjecture so, possibly with a larger constant. 
    \item Are similar results about the probability of random points being co-visible achievable in 3-dimensions? 
\end{itemize}

\section*{Acknowledgments}
We are grateful to Mark de Berg for asking for the perimeter and to Hans Raj Tiwary for pointing out the main obstacle when dealing with the perimeter.


\end{document}